\documentclass[11pt,onecolumn]{IEEEtran}

\usepackage[mathscr]{eucal}
\usepackage[cmex10]{amsmath}
\usepackage{epsfig,epsf,psfrag}
\usepackage{amssymb,amsmath,amsthm,amsfonts,latexsym}
\usepackage{amsmath,graphicx,bm,xcolor,url}
\usepackage[caption=false]{subfig} 
\usepackage{fixltx2e}
\usepackage{array}
\usepackage{verbatim}
\usepackage{bm}
\usepackage{algorithmic, cite}
\usepackage{algorithm}
\usepackage{verbatim}
\usepackage{textcomp}
\usepackage{mathrsfs}
\usepackage{epstopdf}

\catcode`~=11 \def\UrlSpecials{\do\~{\kern -.15em\lower .7ex\hbox{~}\kern .04em}} \catcode`~=13 

\allowdisplaybreaks[1]
 
\newcommand{\nn}{\nonumber}

\newcommand{\calA}{\mathcal{A}}
\newcommand{\calB}{\mathcal{B}}
\newcommand{\calC}{\mathcal{C}}
\newcommand{\calD}{\mathcal{D}}
\newcommand{\calE}{\mathcal{E}}
\newcommand{\calF}{\mathcal{F}}
\newcommand{\calG}{\mathcal{G}}

\newcommand{\calI}{\mathcal{I}}
\newcommand{\calJ}{\mathcal{J}}
\newcommand{\calK}{\mathcal{K}}
\newcommand{\calL}{\mathcal{L}}
\newcommand{\calM}{\mathcal{M}}
\newcommand{\calN}{\mathcal{N}}

\newcommand{\calR}{\mathcal{R}}

\newcommand{\calT}{\mathcal{T}}

\newcommand{\calX}{\mathcal{X}}
\newcommand{\calY}{\mathcal{Y}}

\newcommand{\ba}{\mathbf{a}}
\newcommand{\bA}{\mathbf{A}}

\newcommand{\bD}{\mathbf{D}}

\newcommand{\bG}{\mathbf{G}}

\newcommand{\bI}{\mathbf{I}}
\newcommand{\bj}{\mathbf{j}}

\newcommand{\bR}{\mathbf{R}}

\newcommand{\bu}{\mathbf{u}}
\newcommand{\bU}{\mathbf{U}}
\newcommand{\bv}{\mathbf{v}}
\newcommand{\bV}{\mathbf{V}}

\newcommand{\bx}{\mathbf{x}}
\newcommand{\bX}{\mathbf{X}}
\newcommand{\by}{\mathbf{y}}
\newcommand{\bY}{\mathbf{Y}}

\newcommand{\bZ}{\mathbf{Z}}

\newcommand{\rmd}{\mathrm{d}}

\newcommand{\rme}{\mathrm{e}}

\newcommand{\bbE}{\mathbb{E}}

\newcommand{\bbN}{\mathbb{N}}

\newcommand{\bbR}{\mathbb{R}}

\DeclareMathAlphabet{\mathbsf}{OT1}{cmss}{bx}{n}
\DeclareMathAlphabet{\mathssf}{OT1}{cmss}{m}{sl}

\newcommand{\rvC}{\mathsf{C}}

\newcommand{\rvM}{\mathsf{M}}

\newcommand{\rvV}{\mathsf{V}}

\DeclareSymbolFont{bsfletters}{OT1}{cmss}{bx}{n}  
\DeclareSymbolFont{ssfletters}{OT1}{cmss}{m}{n}
\DeclareMathSymbol{\bsfGamma}{0}{bsfletters}{'000}
\DeclareMathSymbol{\ssfGamma}{0}{ssfletters}{'000}
\DeclareMathSymbol{\bsfDelta}{0}{bsfletters}{'001}
\DeclareMathSymbol{\ssfDelta}{0}{ssfletters}{'001}
\DeclareMathSymbol{\bsfTheta}{0}{bsfletters}{'002}
\DeclareMathSymbol{\ssfTheta}{0}{ssfletters}{'002}
\DeclareMathSymbol{\bsfLambda}{0}{bsfletters}{'003}
\DeclareMathSymbol{\ssfLambda}{0}{ssfletters}{'003}
\DeclareMathSymbol{\bsfXi}{0}{bsfletters}{'004}
\DeclareMathSymbol{\ssfXi}{0}{ssfletters}{'004}
\DeclareMathSymbol{\bsfPi}{0}{bsfletters}{'005}
\DeclareMathSymbol{\ssfPi}{0}{ssfletters}{'005}
\DeclareMathSymbol{\bsfSigma}{0}{bsfletters}{'006}
\DeclareMathSymbol{\ssfSigma}{0}{ssfletters}{'006}
\DeclareMathSymbol{\bsfUpsilon}{0}{bsfletters}{'007}
\DeclareMathSymbol{\ssfUpsilon}{0}{ssfletters}{'007}
\DeclareMathSymbol{\bsfPhi}{0}{bsfletters}{'010}
\DeclareMathSymbol{\ssfPhi}{0}{ssfletters}{'010}
\DeclareMathSymbol{\bsfPsi}{0}{bsfletters}{'011}
\DeclareMathSymbol{\ssfPsi}{0}{ssfletters}{'011}
\DeclareMathSymbol{\bsfOmega}{0}{bsfletters}{'012}
\DeclareMathSymbol{\ssfOmega}{0}{ssfletters}{'012}

\newcommand{\hatm}{\hat{m}}

\newcommand{\tilm}{\tilde{m}}

\newcommand{\bLambda}{\bm{\Lambda}}
\newcommand{\bSigma	}{\bm{\Sigma}}

\newcommand{\hrho}{\hat{\rho}}

\newcommand{\eps}{\varepsilon}

\DeclareMathOperator{\var}{\mathsf{Var}}

\DeclareMathOperator{\cov}{\mathsf{Cov}}

\newcommand{\bzero}{\mathbf{0}}
\newcommand{\bone}{\mathbf{1}}

\newtheorem{theorem}{Theorem} 
\newtheorem{lemma}[theorem]{Lemma}

\newtheorem{proposition}[theorem]{Proposition}

\newtheorem{definition}{Definition}

\newtheorem{remark}{Remark}

\newcommand{\qednew}{\nobreak \ifvmode \relax \else
      \ifdim\lastskip<1.5em \hskip-\lastskip
      \hskip1.5em plus0em minus0.5em \fi \nobreak
      \vrule height0.75em width0.5em depth0.25em\fi}

\usepackage{epsfig} 
\usepackage{epstopdf}  
\usepackage{overpic}

\newcommand{\underg}{\underline{g}}
\newcommand{\overg}{\overline{g}}
\newcommand{\openone}{\leavevmode\hbox{\small1\normalsize\kern-.33em1}} 
\newcommand{\bYtil}{\tilde{\mathbf{Y}}}
\newcommand{\bytil}{\tilde{\mathbf{y}}}
\newcommand{\Ytil}{\tilde{Y}}
\newcommand{\ytil}{\tilde{y}}
 
\begin{document} 
\flushbottom
\title{Second-Order Asymptotics for the  Gaussian \\ MAC with Degraded Message Sets} 

\author{Jonathan Scarlett, {\em Member, IEEE}  $\,$ and  $\,$
        Vincent Y.~F.~Tan, {\em Senior Member, IEEE}  
\thanks{J.~Scarlett was
with the Department of Engineering, University of Cambridge, Cambridge CB2 1PZ, U.K.\ 
He is now with the Laboratory for Information and Inference Systems, \'Ecole Polytechnique F\'ed\'erale de Lausanne, CH-1015, Switzerland
(email:\,jmscarlett@gmail.com).}
\thanks{V.~Y.~F.~Tan was with the Institute for Infocomm Research (I$^2$R), Agency for Science, Technology and Research (A*STAR). He is  now with the   Department of Electrical and Computer Engineering and the Department of Mathematics, National University of Singapore.
(email:\,vtan@nus.edu.sg). }
\thanks{This paper was presented in part at the 2014 IEEE International Symposium on Information Theory in Honolulu, HI.}
}
 
\IEEEpeerreviewmaketitle
 
\maketitle

\begin{abstract} 
This paper studies the second-order asymptotics of the Gaussian multiple-access channel with degraded message sets. For a fixed  average error probability $\eps \in (0,1)$ and an arbitrary point on the boundary of the capacity region, we characterize the speed of convergence of rate pairs that converge to that boundary point for codes that have asymptotic error probability no larger than $\eps$. 
As a stepping stone to this local notion of second-order asymptotics, we study a global notion, and establish relationships between the two.  We provide a numerical example to illustrate how the angle of approach to a boundary point affects the second-order coding rate.  This is the first conclusive characterization of the second-order asymptotics of a network information theory problem in which the capacity region is not a polygon. 
\end{abstract} 

 
\section{Introduction}

In this paper, we revisit the Gaussian multiple-access channel (MAC) with degraded message sets.  This is a communication model in which two independent messages are to be sent from two  sources to a common destination; see Fig.~\ref{fig:gauss_mac}. One encoder, the cognitive or informed encoder, has access to both messages, while the uninformed encoder only has access to its own message.  Both transmitted signals are power limited, and their sum is corrupted by additive white Gaussian noise (AWGN).

The capacity region $\calC$, i.e.\ the set of all pairs of achievable rates, is well-known (e.g.\ see \cite[Ex.~5.18(b)]{elgamal}), and is given by the set of rate pairs $(R_1, R_2)$ satisfying
\begin{align}
R_1 &\le \rvC\big((1-\rho^2) S_1\big) \label{eqn:cap1} \\
R_1 + R_2 &\le \rvC\big( S_1+S_2 + 2\rho\sqrt{ S_1 S_2 }  \big) \label{eqn:cap2}
\end{align} 
for some $\rho \in [0,1]$, where $S_1$ and $S_2$ are the admissible transmit powers, 
and $\rvC(x):=\frac{1}{2}\log(1+x)$ is the Gaussian capacity function. The capacity 
region $\calC$ does not depend on whether the average or maximal error probability  
formalism is employed, and no time-sharing is required. The region  $\calC$ for 
$S_1=S_2=1$ is illustrated in Fig.~\ref{fig:cr}; observe that $\calC$  is formed 
from a union of trapezoids, each parametrized by $\rho$. The vertical line segment 
corresponds to $\rho=0$, while the curved part corresponds to $\rho\in (0,1]$. 
The direct part of the coding theorem for $\calC$ is proved using superposition coding~\cite{cover72}.  

While the capacity region is well-known, there is substantial motivation to understand the {\em second-order asymptotics} for this problem. For any given point $(R_1^*, R_2^*)$ on the boundary of the capacity region, we study the rate of convergence to that point for an $\eps$-reliable code. More precisely, we characterize the set of all $(L_1, L_2)$ pairs, known as second-order coding rates~\cite{Hayashi08,Hayashi09,Nom13,Nom13b}, for which there exist sequences of codes whose asymptotic error probability does not exceed $\eps$, and whose code sizes $M_{1,n}$ and $M_{2,n}$ behave as 
\begin{align}
\log M_{j,n}\ge n R_j^* + \sqrt{n}L_j + o\big( \sqrt{n}\big),\quad j = 1,2. \label{eqn:roughly}
\end{align}
This study allows us to understand the fundamental  tradeoffs between the rates of transmission and average error probability from a perspective different from the study of  error exponents.  Here,  instead of fixing a pair of rates and studying the exponential decay of the  error probability $\eps$, we fix $\eps$ and study the speed at which a sequence of rate pairs approaches an information-theoretic limit as the blocklength grows. 

\begin{figure}[t]
\centering
\setlength{\unitlength}{.05cm}
\begin{picture}(200,100)

\put(0, 15){\vector(1, 0){30}}
\put(0, 75){\vector(1, 0){30}}
\put(0, 15){\vector(1, 2){30}}

\put(30, 0){\line(1, 0){30}}
\put(30, 30){\line(1, 0){30}}
\put(30, 0){\line(0,1){30}}
\put(60, 0){\line(0, 1){30}}

\put(30, 60){\line(1, 0){30}}
\put(30, 90){\line(1, 0){30}}
\put(30, 60){\line(0,1){30}}
\put(60, 60){\line(0, 1){30}}

\put(60, 75){\vector(1, -1){30}}
\put(60, 15){\vector(1, 1){30}}

\put(97, 45){\circle{14}}

\put(94, 43){\mbox{$+$}}

\put(97, 82){\vector(0, -1){30}}

\put(104, 45){\vector(1, 0){30}}

\put(134, 30){\line(1, 0){30}}
\put(134, 60){\line(1, 0){30}}
\put(134, 30){\line(0,1){30}}
\put(164, 30){\line(0, 1){30}}

\put(164, 45){\vector(1, 0){30}}

\put(-13,13){\mbox{$\rvM_2$}}
\put(-13,73){\mbox{$\rvM_1$}}

\put(75,61){\mbox{$\bX_1$}}
\put(75,24){\mbox{$\bX_2$}}

\put(39,73){\mbox{$f_{1,n}$}}
\put(39,13){\mbox{$f_{2,n}$}}

\put(90,85){\mbox{$\bZ\sim\calN(\bzero,\bI_n)$}}

\put(115,48){\mbox{$\bY$}}

\put(144,45){\mbox{$\varphi_n$}}

\put(175,49){\mbox{$(\hat{\rvM}_1, \hat{\rvM}_2)=\varphi_n(\bY)$}}
\end{picture} 
\caption{The model for the Gaussian MAC with degraded message sets. }
\label{fig:gauss_mac}
\end{figure}
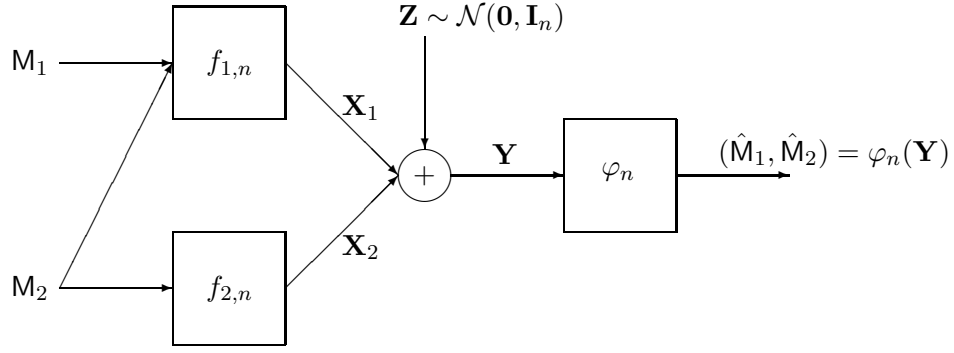

\subsection{Related Work}
The most notable early work on the second-order asymptotics for channel coding is that of Strassen~\cite{Strassen}, who considered discrete memoryless channels. For the single-user AWGN channel with a maximal power constraint $S$, a specialization of our model with $M_{2,n}=1$, Hayashi~\cite{Hayashi09} and Polyanskiy {\em et al.}~\cite{PPV10} showed that the optimum (highest) second-order coding rate is $\sqrt{\rvV(S) }\Phi^{-1}(\eps)$, where $\rvV(x):=\frac{x(x+2)}{2(x+1)^2}$ is the {\em Gaussian dispersion function}. 
Polyanskiy {\em et al.}~\cite[Thm.~54]{PPV10} and Tan-Tomamichel \cite{TanTom14} showed the refined asymptotic expansion
\begin{equation}
\log M^*(n,\eps)=n\rvC(S) + \sqrt{n\rvV(S)}\Phi^{-1}(\eps)+\frac{1}{2}\log n+O(1), \label{eqn:single_user}
\end{equation}
where $M^*(n,\eps)$ is the maximum size  of a length-$n$ block code  with average  error probability not exceeding $\eps$.
In fact, the expression for $\rvV(S)$ was already known to Shannon~\cite[Sec.~X]{Sha59b}, who analyzed  the reliability function of the AWGN channel for rates close to  capacity.  

There have been numerous attempts to study the finite blocklength behavior and   second-order asymptotics for MACs~\cite{TK12, huang12, Mol12, Mol12b, Mol13, Ver12,Mou13,  Scarlett13b,  Haim12}, but most of these works focus on inner bounds (the direct part). The development of tight and easily-evaluated converse bounds remains more modest, and those available do not match the direct part in general or are very restrictive (e.g.\ product channels were considered in \cite{Haim12}). We will see that the assumption of Gaussianity of the channel model together with the degradedness of the message sets allows us to circumvent some of the difficulties in proving second-order converses for the MAC, thus allowing us to obtain a {\em conclusive}  second-order result. 

We focus primarily on {\em local} second-order asymptotics propounded by Haim {\em et al.}~\cite{Haim12} for general network information theory problems, where a boundary point is fixed and the rate of approach is characterized. This is different from the global asymptotics studied in~\cite{TK12, huang12, Mol12,Mol12b,Mol13,Ver12,Mou13,Scarlett13b}, which we also study here as an initial step towards obtaining the local result. 
 
\subsection{Main Contributions}
Our main contribution is the characterization of the set of admissible local second-order coding rates 
$(L_1, L_2)$ for points on the curved part of the boundary of the capacity region (Theorem \ref{thm:local}). 
For a point characterized by $\rho \in (0,1)$, we show that the achievable second-order rate pairs $(L_1, L_2)$ are precisely those satisfying
\begin{align}
\begin{bmatrix}
L_1\\ L_1+L_2
\end{bmatrix}\in \bigcup_{\beta\in\bbR} \big\{ \beta\, \bD(\rho)+\Psi^{-1}(\bV(\rho),\eps)  \big\}, \label{eqn:L_intro}
\end{align}
where the entries of $\bD(\rho)$ are the derivatives of the capacities in \eqref{eqn:cap1}--\eqref{eqn:cap2}, $\bV(\rho)$ is the {\em dispersion matrix}~\cite{TK12,huang12}, and $\Psi^{-1}$ is the $2$-dimensional generalization of the inverse of the cumulative distribution function of a Gaussian. (All quantities are defined precisely in the sequel.) Thus, the contribution from the Gaussian approximation $\Psi^{-1}(\bV(\rho),\eps)$ is insufficient for characterizing the second-order asymptotics of multi-terminal channel coding problems in general; in this case, the vector $\bD(\rho)$ is also required.  
This is in stark contrast to single-user problems (e.g.~\cite{Hayashi08, Hayashi09, PPV10, Strassen, Nom13b}) and the (two-encoder) Slepian-Wolf problem~\cite{TK12, Nom13} where the Gaussian approximation in terms of a dispersion quantity is sufficient for the second-order asymptotics.
Our main result, which comprises the statement in \eqref{eqn:L_intro}, provides the first complete characterization of the local second-order asymptotics of a multi-user information theory problem in which the boundary of the capacity region (or optimal rate region for source coding problems) is curved.

Some intuition can be gained as to why the extra derivative term is needed by considering the possible
angles of approach to a fixed boundary point $(R_1^*,R_2^*)\in\calC$.  Using a {\em single} multivariate 
Gaussian input distribution with correlation $\rho$ for all blocklengths is suboptimal in the second-order 
sense, as we can only achieve the angles of approach within the trapezoid parametrized by  $\rho$ (see 
Fig.~\ref{fig:cr} and its caption). Our strategy is to consider sequences of input distributions that vary 
with the blocklength, i.e.\ they are parametrized by a sequence $\{\rho_n\}_{n\in\bbN}$ that converges to 
$\rho$ with speed $\Theta\big(\frac{1}{\sqrt{n}}\big)$. By a Taylor expansion of the first-order capacity vector 
$\bI(\rho)$ (the vector of capacities in \eqref{eqn:cap1}--\eqref{eqn:cap2}), 
\begin{equation}
\bI(\rho_n)\approx\bI(\rho) + (\rho_n-\rho) \bD(\rho), \label{eqn:back_envelope}
\end{equation}
we see that this sequence results in the derivative/slope term $\bD(\rho)$ 
observed in~\eqref{eqn:L_intro}.   Thus, the slope term corresponds 
to the deviation of $\rho_n$ from $\rho$, while the dispersion term involving $\bV(\rho)$ 
results from, by now, standard central limit (fixed error) analysis of Shannon-theoretic 
coding problems~\cite{Tan_FnT}.
 
We briefly comment on $\rho_n$ converging to $\rho$ at different speeds.
If $\rho_n-\rho=o(\frac{1}{\sqrt{n}})$, then the contribution of the remainder
term in \eqref{eqn:back_envelope} is dominated by the dispersion term, and 
hence this is, up to second order, equivalent to considering $\rho_n = \rho$.
In contrast, for $\rho_n-\rho=\omega(\frac{1}{\sqrt{n}})$, 
this remainder term dominates the dispersion term. 
Nevertheless, this case does not feature in the local result, due to the way we define the 
second-order coding rate region in \eqref{eqn:roughly}---the backoff terms with 
coefficients $L_1$ and $L_2$ scale as $\sqrt{n}$.  In particular, we show in the
converse proof that if $\rho_n-\rho=\omega(\frac{1}{\sqrt{n}})$, then no finite 
$(L_1,L_2)$ pairs  satisfy the conditions in this definition.


\begin{figure}[t]
\centering
\begin{overpic}[width = .553\columnwidth]{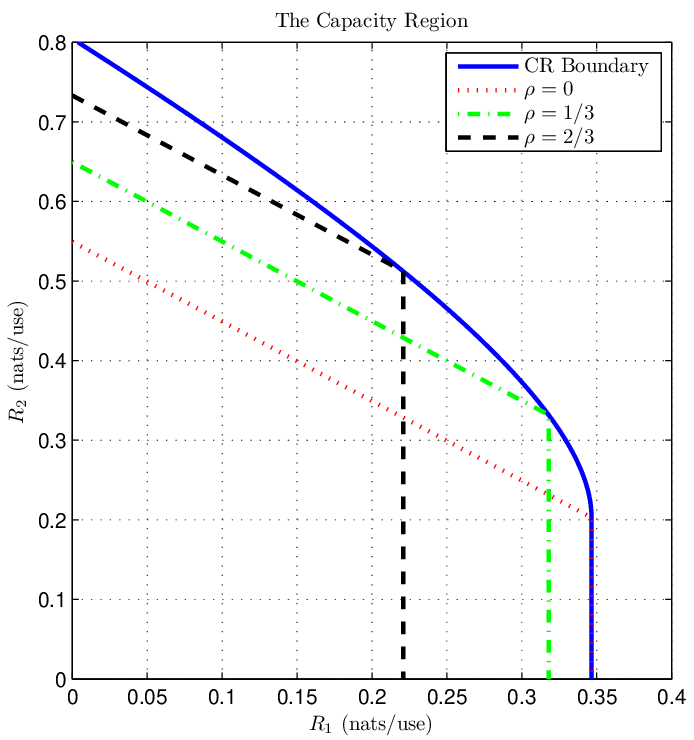}  
\thicklines 
\put(40, 48){\vector(1, 1){15}}
\put(65, 42.5){\vector(-1, 2){10}}  
\put(38, 45){\mbox{$\bv$}}
\put(65, 39.5){\mbox{$\bv'$}}
\end{overpic} 
\caption{Capacity region of the Gaussian MAC with degraded message sets in the case that 
$S_1 =  S_2 = 1$. Observe that $\rho \in [0,1]$ parametrizes points on the boundary.    
The vertical line segment corresponds to $\rho=0$, while the curved part corresponds to 
$\rho\in (0,1]$.  Each $\rho \in (0,1]$ corresponds to a trapezoid of rate pairs that are 
achievable by a unique input distribution   $\calN(\bzero,\bSigma(\rho))$. This 
coding strategy is insufficient to allow for all possible angles of approach to the fixed point 
parametrized by $\rho$, as there are non-empty regions within $\calC$ that not in the 
trapezoid parametrized by $\rho$. In the figure above with $\rho=\frac{2}{3}$, one can 
approach the corner point in the direction indicated by the vector $\bv$ using the fixed 
input distribution $\calN(\bzero,\bSigma(\frac{2}{3}))$, but the same is not true of
the direction indicated by $\bv'$, since the approach is from outside the trapezoid.
}
\label{fig:cr} 
\end{figure}
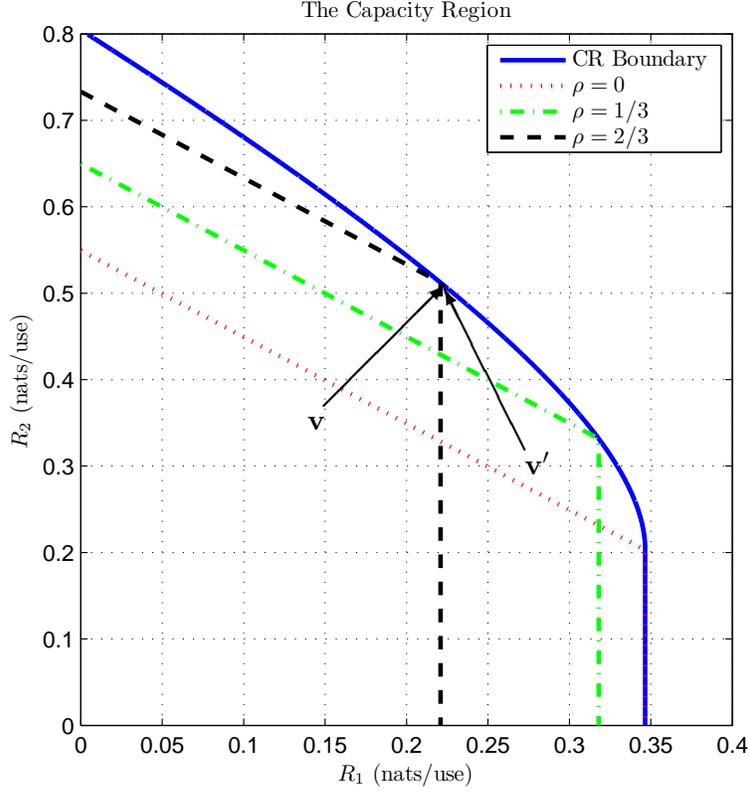

  An auxiliary contribution is  a global second-order result~\cite{Haim12,TK12}  (Theorem \ref{thm:global}), 
which we  use  as an important {\em stepping stone} to obtain our local second-order result.   
We show   
that for any sequence $\rho_n\in[0,1]$, all rate pairs $(R_{1,n}, R_{2,n})$ satisfying
\begin{equation}
\begin{bmatrix}
R_{1,n}\\ R_{1,n}+R_{2,n}
\end{bmatrix}    \in\bI(\rho_n) + \frac{\Psi^{-1}(\bV(\rho_n),\eps)}{\sqrt{n}} + o\bigg(\frac{1}{\sqrt n}\bigg)\bone \label{eq:global_example}
\end{equation}
are achievable at blocklength $n$ and with average error probability no 
larger than $\eps + o(1)$.  Our proof technique yields a third-order 
term that remains $o( \frac{1}{\sqrt n})$ no matter how $\rho_n$ 
varies with $n$. This property does not typically hold in previous 
results on multi-user fixed error asymptotics, but it
turns out to be crucial in deriving the local result and the 
additional slope term (cf.~\eqref{eqn:back_envelope}), at least
using our proof techniques.

In summary, we submit that both the global and local results on their 
own provide complementary and useful insights into fundamental limits 
of the communication system, but in this paper our main goal is the latter.


\section{Problem Setting and Definitions}
In this section, we state the channel model, various definitions and some known results.   

\subsubsection*{Notation} Given  integers $l\le m$, we use the discrete interval~\cite{elgamal} notations $[l:m]:=\{l,\ldots, m\}$ and   $[m] :=[1:m]$. All $\log$'s  and $\exp$'s are  with respect to the  natural base $\rme$.  The $\ell_p$-norm of the vectorized version of matrix $\bA$ is denoted by $\|\bA\|_p := \big(\sum_{i,j} |a_{i,j}|^p\big)^{1/p}$. For two vectors of the same length $\ba,\mathbf{b}\in\bbR^d$, the notation $\ba\le\mathbf{b}$ means that $a_j\le b_j$ for all $j \in [d]$. The notation $\calN(\bu;\bm{\mu}, \bLambda)$   denotes the  multivariate Gaussian probability density function (pdf) with mean $\bm{\mu}$ and covariance $\bLambda$. The argument $\bu$ will often be omitted. We use standard asymptotic notations: $f_n\in O(g_n)$ if and only if  (iff) $\limsup_{n\to\infty} \big|f_n/g_n\big|<\infty$; $f_n \in \Omega(g_n)$ iff $g_n \in O(f_n)$; $f_n \in \Theta(g_n)$ iff $f_n \in O(g_n)\cap\Omega(g_n)$; $f_n \in o(g_n)$ iff $\limsup_{n\to\infty} \big|f_n/g_n\big|=0$; and  $f_n \in \omega(g_n)$ iff $\liminf_{n\to\infty} \big|f_n/g_n\big|=\infty$. 
\subsection{Channel Model}
The signal model is given by 
\begin{align}
Y=  X_1 + X_2 + Z, \label{eqn:channel}
\end{align}
where $X_1$ and $X_2$ represent the inputs to the channel,  $Z\sim\calN(0,1)$ is additive Gaussian noise with mean zero and unit variance, and $Y$ is the output of the channel. Thus, the   channel  from $(X_1, X_2)$ to $Y$ can be written as 
\begin{align}
W(y|x_1, x_2) = \frac{1}{\sqrt{2\pi}}\exp\left(-\frac{1}{2}(y-x_1-x_2)^2 \right).
\end{align}
The channel is used $n$ times in a memoryless manner without feedback. The channel inputs (i.e., the transmitted codewords) $\bx_1 = (x_{11},\ldots, x_{1n})$ and $\bx_2= (x_{21},\ldots, x_{2n})$ are required to  satisfy the maximal power constraints  
\begin{align}
\|\bx_1\|_2^2\le n S_1,\quad\mbox{and}\quad  \|\bx_2\|_2^2\le n S_2, \label{eqn:power_constraints}
\end{align}
where  $S_1$ and $S_2$    are arbitrary positive   numbers. We do not incorporate multiplicative gains $g_1$ and $g_2$ to $X_1$ and $X_2$ in the channel model in~\eqref{eqn:channel}; this is without loss of generality, since in the presence of these gains we may equivalently redefine~\eqref{eqn:power_constraints} with $S_j' :=S_j / g_j^2$   for $j = 1,2$. 

\subsection{Definitions}
\begin{definition}[Code] \label{def:code}
An  {\em $(n,M_{1,n},M_{2,n},S_1, S_2,\eps_n)$-code}  for the Gaussian MAC with degraded message sets consists of two encoders $f_{1,n},f_{2,n}$ and a decoder $\varphi_n$ of the form  $f_{1,n} : [M_{1,n}]\times [M_{2,n}] \to \bbR^n$,      $f_{2,n} : [M_{2,n}]\to\bbR^n$ and    $\varphi_n : \bbR^n\to [M_{1,n}]\times [M_{2,n}]$ satisfying
\begin{align}
  \|f_{1,n}(m_1, m_2)\|_2^2 &\le n S_1    \quad\forall  \, (m_1,m_2)\in [M_{1,n}]\times [M_{2,n}] , \label{eqn:power1} \\
 \|f_{2,n}(m_2)\|_2^2&\le n S_2    \quad\forall \,   m_2 \in  [M_{2,n}]   \label{eqn:power2}, \\
\Pr\big( (\rvM_1,\rvM_2)\ne (\hat{\rvM}_1,\hat{\rvM}_2) \big)  &\le\eps_n , \label{eqn:error_prob1}
\end{align}
where the messages $\rvM_1$ and $\rvM_2$ are uniformly distributed on $[M_{1,n}]$ and $[M_{2,n}]$ respectively, and $(\hat{\rvM}_1,\hat{\rvM}_2) :=\varphi_n(Y^n)$ is the decoded message pair. 
\end{definition}
Since $S_1$ and $S_2$ are fixed positive numbers, we suppress the dependence of the subsequent definitions, results and parameters on these constants. We will often make reference to {\em $(n,\eps)$-codes}; this is the family of $(n,M_{1,n},M_{2,n},S_1, S_2,\eps)$-codes where the sizes $M_{1,n},M_{2,n}$ are left unspecified.

\begin{definition}[$(n,\eps)$-Achievability] \label{def:neps}
A pair of  non-negative  numbers $(R_1, R_2)$ is {\em $(n,\eps)$-achievable} if there exists an  $(n,M_{1,n},M_{2,n},S_1, S_2,\eps_n)$-code such that 
\begin{align}
\frac{1}{n}\log M_{j,n}\ge R_j,\quad j = 1,2,\quad\mbox{and}\quad \eps_n \le \eps.
\end{align}
The {\em $(n,\eps)$-capacity region} $\calC(n,\eps) \subset\bbR_+^2$ is defined to be the set of all $(n,\eps)$-achievable rate pairs $(R_1, R_2)$. 
\end{definition}

Definition \ref{def:neps} is a non-asymptotic one that is used primarily for the global second-order results. We now introduce asymptotic-type definitions that involve the existence of {\em sequences} of codes. 
\begin{definition}[First-Order Coding Rates]
A pair of  non-negative numbers $(R_1, R_2)$ is {\em $\eps$-achievable} if there exists a sequence of $(n,M_{1,n},M_{2,n},S_1, S_2,\eps_n)$-codes such that  
\begin{align}
\liminf_{n\to\infty}\frac{1}{n}\log M_{j,n} \ge R_j,\quad j = 1,2,\quad\mbox{and}\quad \limsup_{n\to\infty}\eps_n \le \eps.
\end{align}
The {\em $\eps$-capacity region} $\calC(\eps) \subset\bbR_+^2$ is defined to be the closure of the set of all $\eps$-achievable rate pairs $(R_1, R_2)$.  The {\em capacity region}  $\calC$ is defined as 
\begin{equation}
\calC:=\bigcap_{\eps > 0}\calC(\eps)=\lim_{\eps\to 0}\calC(\eps),
\end{equation}
where the   limit exists because of the monotonicity of $\calC(\eps)$.
\end{definition}

Next, we state the most important definitions concerning local second-order coding rates in the spirit of Nomura-Han~\cite{Nom13} and Tan-Kosut~\cite{TK12}.  We will spend the majority of the paper developing tools to characterize these rates. Here    $(R_1^*, R_2^*)$ is a pair of rates   on the boundary of $\calC(\eps)$. 

\begin{definition}[Second-Order Coding Rates] \label{def;second}
A pair of numbers $(L_1, L_2)$ is {\em $(\eps,R_1^*,R_2^*)$-second-order achievable} if there exists a sequence of $(n,M_{1,n},M_{2,n},S_1, S_2,\eps_n)$-codes such that  
\begin{align}
\liminf_{n\to\infty}\frac{1}{\sqrt{n}}(\log M_{j,n} - nR_j^*) \ge L_j,\quad j = 1,2,\quad\mbox{and}\quad \limsup_{n\to\infty}\eps_n \le \eps. \label{eqn:second_def}
\end{align}
The {\em $(\eps,R_1^*,R_2^*)$-optimal second-order coding rate region} $\calL(\eps;R_1^*,R_2^*) \subset\bbR^2$ is defined to be the closure of the  set of all $(\eps,R_1^*,R_2^*)$-second-order achievable rate pairs  $(L_1, L_2)$. 
\end{definition}
Stated differently, if $(L_1, L_2)$ is   $(\eps,R_1^*,R_2^*)$-second-order  achievable, then there are codes  whose error probabilities are asymptotically no larger than $\eps$, and whose sizes $(M_{1,n},M_{2,n})$ satisfy the asymptotic relation in~\eqref{eqn:roughly}. Even though we refer to  $L_1$ and $L_2$ as ``rates'', they may be negative~\cite{Hayashi08,Hayashi09,Nom13,Nom13b}.  A negative value corresponds to a backoff from the first-order term, whereas a positive value corresponds to an addition to the first-order term.

\subsection{Existing  First-Order Results}
To put things in context, we review some existing results concerning the $\eps$-capacity region. To state the result compactly, we define the {\em mutual information (or capacity) vector} as 
\begin{align}
\bI(\rho)=\begin{bmatrix}
I_1(\rho) \\ I_{12}(\rho)
\end{bmatrix} := \begin{bmatrix}
\rvC\big( S_1(1-\rho^2)  \big)  \\ \rvC\big( S_1+S_2 + 2\rho\sqrt{S_1S_2}   \big) 
\end{bmatrix}   \label{eqn:mi_vec}
\end{align}
where $\rho\in [-1,1]$. 
For a pair of rates $(R_1, R_2)$, let the {\em rate vector} be
\begin{align}
\bR := \begin{bmatrix}
R_1 \\ R_1 + R_2
\end{bmatrix} . \label{eqn:rate_vec}
\end{align}
A statement of the following result is provided in~\cite[Ex.~5.18(b)]{elgamal}. A weak converse was proved for the more general Gaussian MAC a with common message in \cite{BLW12}.
\begin{proposition}[Capacity Region] \label{prop:first}
The capacity region of the Gaussian MAC with degraded message sets is given by
\begin{align}
\calC= \bigcup_{0\le\rho\le 1}\left\{ (R_1, R_2 ) \in\bbR_+^2 : \bR \le \bI(\rho) \right\}   . \label{eqn:first_order}
\end{align}
\end{proposition}  
The union on the right is a subset of $\calC(\eps)$ for every $\eps\in (0,1)$. However, only the weak converse is implied by~\eqref{eqn:first_order}. The strong converse has not been demonstrated previously. Thus, a by-product of the derivation of the second-order asymptotics in this paper is the strong converse, allowing us to assert that  for all $\eps\in (0,1)$, 
\begin{align}
\calC=\calC(\eps). \label{eqn:str_conv}
\end{align}
The direct part of Proposition~\ref{prop:first} can be proved using superposition coding~\cite{cover72}, treating $X_2$ as the cloud center and $X_1$ as the satellite codeword. The input distribution to achieve a point on the boundary characterized by some $\rho \in [0,1]$ is a  $2$-dimensional Gaussian   with mean zero and covariance matrix 
\begin{align} \label{eqn:sigmamatrix}
\bSigma(\rho) :=\begin{bmatrix}
S_1 & \rho \sqrt{S_1S_2}\\
\rho \sqrt{S_1S_2} & S_2
\end{bmatrix} .
\end{align}
Thus, the parameter $\rho$ represents the correlation between the two users' codewords.

\section{Global Second-Order Results}\label{sec:global}
In this section, we present inner and outer bounds on $\calC(n,\eps)$.  We begin with some definitions. Let $\rvV(x,y) := \frac{x(y+2)}{2(x+1)(y+1)}$ be the {\em Gaussian cross-dispersion function} and let $\rvV(x) := \rvV(x,x)$ be the  {\em Gaussian dispersion function}~\cite{Sha59b,PPV10,Hayashi09} for a single-user AWGN channel with signal-to-noise ratio $x$.  For fixed $0\le \rho\le 1$, define the {\em information-dispersion matrix}
\begin{align}
\bV(\rho) :=\begin{bmatrix}
V_1(\rho)  & V_{1,12}(\rho)\\ V_{1,12}(\rho)  & V_{12}(\rho) \end{bmatrix}, \label{eqn:inf_disp_matr}
\end{align}
where the elements of the matrix are  
\begin{align}
V_1(\rho) &:= \rvV\big( S_1 (1-\rho^2)\big) ,   \\ 
V_{1,12}(\rho) &:= \rvV\big( S_1 (1-\rho^2), S_1 + S_2 + 2\rho\sqrt{S_1 S_2} \big)  , \\
V_{12}(\rho) &:= \rvV\big(S_{1}+S_{2}+2\rho\sqrt{S_{1}S_{2}} \big)  . 
\end{align}
Let  $(X_1,X_2)\sim P_{X_1, X_2} = \calN(\bzero;\bSigma(\rho))$,   and define $Q_{Y|X_2}$ and $Q_Y$ to be Gaussian distributions   induced by $P_{X_1, X_2}$ and the channel $W$, namely
\begin{align}
Q_{Y|X_2}(y|x_2) & : = \calN\big(y; x_2 (1+ \rho \sqrt{S_1/S_2}) , 1+ S_1(1- \rho ^2) \big),  \label{eqn:out1}\\
Q_{Y}(y) & : = \calN\big(y;0, 1+S_1 + S_2 + 2\rho \sqrt{S_1 S_2}\big). \label{eqn:out2}
\end{align}  
It should be noted that the random variables  $(X_1, X_2)$ and the densities $Q_{Y|X_2}$ and $Q_Y$  all depend on $\rho$; this dependence is suppressed throughout the paper.  The mutual information vector  $\bI(\rho)$ and  information-dispersion matrix $\bV(\rho)$ are  the mean vector and conditional covariance matrix of the information density vector 
\begin{align} 
\bj(X_1, X_2  , Y) := \begin{bmatrix}
j_1(X_1, X_2  , Y) \\ j_{12}(X_1, X_2  , Y)
\end{bmatrix}  = \begin{bmatrix}
\log\dfrac{W(Y|X_1,X_2)}{Q_{Y|X_2}(Y|X_2)} , & \log\dfrac{W(Y|X_1,X_2)}{Q_{Y}(Y)}  
\end{bmatrix}^T.
\label{eqn:info_dens}
\end{align} 
That is, we can write $\bI(\rho)$ and $\bV(\rho)$ as
\begin{align}
\bI(\rho)   &=\bbE\big[\, \bj(X_1,X_2,Y)\big],  \label{eqn:mean_v}\\
  \bV(\rho)   &=\bbE\big[\cov \big(\bj(X_1,X_2,Y) \, \big|\, X_1,X_2 \big)\big].\label{eqn:cov_v}
\end{align}

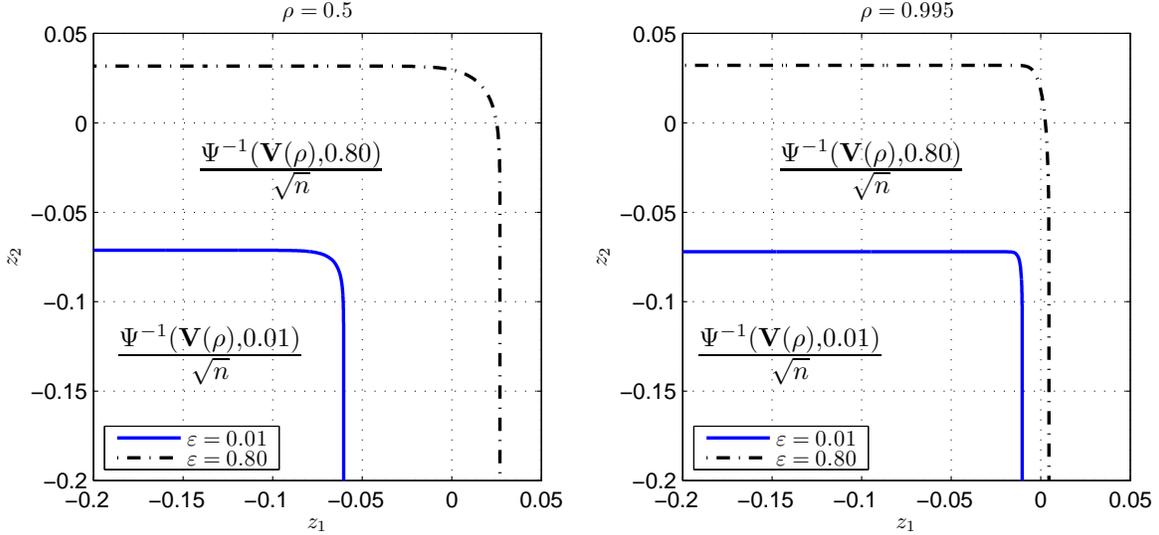
\begin{figure}
\centering
\begin{overpic}[width =.85\columnwidth]{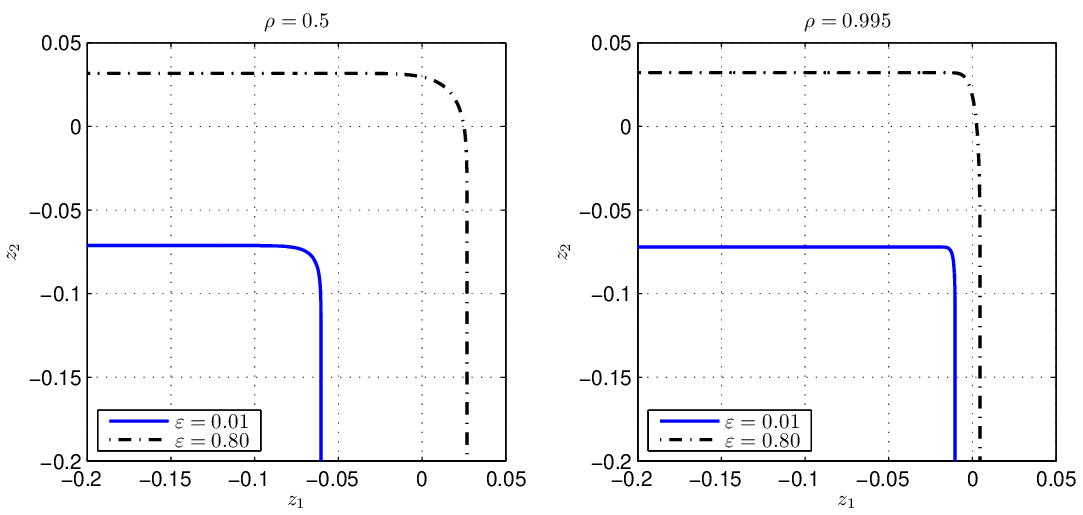}
\put(17,31){\Large $\frac{\Psi^{-1}(\bV(\rho),0.80)}{\sqrt{n}}$ }
\put(67,31){\Large $\frac{\Psi^{-1}(\bV(\rho),0.80)}{\sqrt{n}}$ } 
\put(10,15){\Large $\frac{\Psi^{-1}(\bV(\rho),0.01)}{\sqrt{n}}$ }
\put(60,15){\Large $\frac{\Psi^{-1}(\bV(\rho),0.01)}{\sqrt{n}}$ }
\end{overpic}
  \caption{Illustration of the set $\Psi^{-1}(\bV(\rho),\eps)/\sqrt{n}$ with $n=500$, $S_1=S_2=1$,  $\rho=0.5$ (moderate correlation) and $\rho=0.995$ (high correlation). The information dispersion matrix $\bV(\rho)$ is defined in \eqref{eqn:inf_disp_matr}. In the plots, $\eps$ takes two values, $0.01$ and $0.80$. }
  \label{fig:psi_set}
\end{figure}

For a given point $(z_1, z_2) \in \bbR^2$ and a (non-zero) positive semi-definite matrix $\bV$, define
\begin{align}
\Psi(z_1, z_2;\bV) :=\int_{-\infty}^{z_2}\int_{-\infty}^{z_1}\calN(\bu;\bzero,\bV)\,\rmd \bu,
\end{align}
and for a given $\eps\in (0,1)$, define the set
\begin{align}
\Psi^{-1}(\bV,\eps):=\left\{ (z_1, z_2)\in\bbR^2:\Psi(-z_1,- z_2;\bV) \ge  1-\eps\right\}. \label{eqn:psiinv}
\end{align}
These quantities can be thought of as the generalization of the cumulative distribution function (cdf) of the standard Gaussian $\Phi (z) := \int_{-\infty}^z \calN(u;0,1)\, \rmd u$ and its   inverse $\Phi^{-1}(\eps):=\sup\big\{z\in\bbR :\Phi(-z)\ge 1-\eps\big\}$ to the bivariate case.  For $\eps< \frac{1}{2}$, the points contained in  $\Psi^{-1}(\bV,\eps)$ have negative coordinates.  See Fig.~\ref{fig:psi_set} for an illustration of (scaled versions of) $\Psi^{-1}(\bV(\rho),\eps)$.

Let $\underg(\rho,\eps,n)$ and $\overg(\rho,\eps,n)$ be arbitrary functions of $\rho$, $\eps$ and $n$ for now, and define the inner and outer regions
\begin{align}
\calR_{\mathrm{in}}(n,\eps;\rho)  &:= \bigg\{ (R_1,R_2) \in \bbR^2 \,:\, \bR \in \bI(\rho) + \frac{\Psi^{-1}(\bV(\rho),\eps)}{\sqrt{n}} +  \underg(\rho,\eps,n) \bone\bigg\}, \label{eqn:Rin}   \\
\calR_{\mathrm{out}}(n,\eps;\rho)& := \bigg\{ (R_1,R_2) \in \bbR^2 \,:\, \bR \in \bI(\rho) + \frac{\Psi^{-1}(\bV(\rho),\eps)}{\sqrt{n}} +  \overg(\rho,\eps,n) \bone\bigg\} \label{eqn:Rout} .
\end{align}

\begin{theorem}[Global Bounds on the $(n,\eps)$-Capacity Region] \label{thm:global}
There exist functions $\underg(\rho,\eps,n)$ and $\overg(\rho,\eps,n)$ such that the $(n,\eps)$-capacity region satisfies
\begin{align}
\bigcup_{0\le\rho\le 1}\calR_{\mathrm{in}}(n,\eps;\rho) \subset\calC(n,\eps) \subset\bigcup_{-1\le\rho\le 1}\calR_{\mathrm{out}}(n,\eps;\rho), \label{eqn:unions}
\end{align}
and such that $\underg$ and $\overg$ satisfy the following properties:
\begin{enumerate}
  \item For any $\eps\in(0,1)$ and any sequence $\{\rho_n\}$ converging to some value $\rho \notin \{-1,+1\}$, we have
        \begin{equation}
            \underg(\rho_n,\eps,n) = O\left(\frac{\log n}{n}\right),\quad\mbox{and}\quad \overg(\rho_n,\eps,n) = O\left(\frac{\log n}{n}\right). \label{eqn:thirdorder1}
        \end{equation}
  \item For any $\eps\in(0,1)$ and any sequence $\{\rho_n\}$ with $\rho_n\to\rho\in\{ -1,+1\}$, we have
        \begin{equation}
            \underg(\rho_n,\eps,n) = o\left(\frac{1}{\sqrt{n}}\right),\quad\mbox{and}\quad \overg(\rho_n,\eps,n) = o\left(\frac{1}{\sqrt{n}}\right). \label{eqn:thirdorder2}
        \end{equation}
\end{enumerate}
\end{theorem}
The proof of Theorem~\ref{thm:global} is provided in Section~\ref{sec:prf_global}. We remark that even though the union for the outer bound is taken over $\rho\in [-1,1]$, only the values $\rho\in[0,1]$ will play a role in establishing the local asymptotics in Section~\ref{sec:local}, since negative values of $\rho$ are not even first-order optimal, i.e. they fail to achieve a point on the boundary of the capacity region.

Note that we do not claim the remainder terms in \eqref{eqn:thirdorder1}--\eqref{eqn:thirdorder2} to be uniform
in $\{\rho_n\}$; such uniformity will not be required in establishing our main
local result below.  On the other hand, it is crucial that values of $\rho$ varying with $n$ are 
handled (in contrast, most existing global results in other settings consider fixed input distributions).

\section{Local Second-Order Coding Rates}\label{sec:local}
In this section, we present our main result, namely, the characterization of the $(\eps,R_1^*,R_2^*)$-optimal second-order coding rate region $\calL(\eps;R_1^*,R_2^*)$ (see Definition \ref{def;second}), where $(R_1^*,R_2^*)$ is an arbitrary point on the boundary of $\calC$. Our result is stated in terms of the derivative of the mutual information vector  with respect to $\rho$, namely 
\begin{align}
\bD(\rho)=\begin{bmatrix}
D_1(\rho) \\ D_{12}(\rho) 
\end{bmatrix} :=\frac{\rmd}{\rmd\rho}\begin{bmatrix}
I_1(\rho) \\ I_{12}(\rho) 
\end{bmatrix},  \label{eqn:derI} 
\end{align}
where the individual derivatives are given by
\begin{align}
\frac{\rmd I_1(\rho)}{\rmd\rho}  & =  \frac{-S_1\rho}{1+S_1(1-\rho^2)},\label{eqn:dvalues0} \\
 \frac{\rmd I_{12}(\rho)}{\rmd\rho}  & = \frac{\sqrt{S_1 S_2}}{1+S_1 + S_2 + 2\rho\sqrt{S_1S_2}} . \label{eqn:dvalues}
\end{align}
For a vector $\bv=(v_1, v_2)\in\bbR^2$, we define the {\em down-set} of $\bv$ as
\begin{equation}
\bv^- := \{(w_1,  w_2) \in \bbR^2 : w_1 \le v_1, w_2\le v_2\}. \label{eqn:minus_notation}
\end{equation}

\begin{theorem}[Optimal Second-Order Coding Rate Region] \label{thm:local}
Depending on $(R_1^*, R_2^*)$, we have the following three cases:
\begin{enumerate}
\item[(i)] If   $R_1^*=I_1(0)$ and  $R_1^*+R_2^* \le I_{12}(0)$ (vertical segment of the boundary corresponding to  $\rho=0$), then
\begin{align}
\calL(\eps;R_1^*,R_2^*) = \left\{ (L_1, L_2) \in \bbR^2 : L_1 \le  \sqrt{V_1(0)} \Phi^{-1}(\eps)\right\} \label{eqn:second1} .
\end{align}
\item[(ii)]  If $R_1^*=I_1(\rho)$ and $R_1^*+R_2^*=I_{12}(\rho)$ (curved segment of the boundary corresponding to  $0 <\rho <1$), then
\begin{align}
\calL(\eps;R_1^*,R_2^*) =    \left\{ (L_1, L_2) \in \bbR^2 : \begin{bmatrix}
L_1 \\ L_1+ L_2 
\end{bmatrix} \in  \bigcup_{\beta\in \bbR}\bigg\{  \beta\, \bD(\rho) +\Psi^{-1}(\bV(\rho),\eps)  \bigg\}\right\}. \label{eqn:second2}
\end{align}
\item[(iii)]  If   $R_1^*=0$ and  $R_1^*+R_2^* = I_{12}(1)$ (point on the vertical axis corresponding to  $\rho=1$), then
\begin{equation}
    \calL(\eps;R_1^*,R_2^*) =     \left\{(L_1,L_2)\in\bbR^2 \,:\, \begin{bmatrix}
L_1 \\ L_1+ L_2 
\end{bmatrix} \in  \bigcup_{\beta\le0}   \bigg\{ \beta\, \bD(1)  + \begin{bmatrix}
0 \\ \sqrt{ V_{12}(1)}\Phi^{-1}(\eps)
\end{bmatrix}^-\bigg\}    \right\} .\label{eqn:second3}
\end{equation}
\end{enumerate}
\end{theorem}
The proof of Theorem~\ref{thm:local} is provided in Section~\ref{sec:prf_local}. It leverages on the global second-order result in Theorem~\ref{thm:global}.

\subsection{Discussion} \label{sec:discuss}
Observe that in case (i), the second-order region is simply characterized by a scalar dispersion term $V_1(0)$ and the inverse of the Gaussian cdf $\Phi^{-1}$.  Roughly speaking, in this part of the boundary, there is effectively only a single rate constraint in terms of $R_1$, since we are operating ``far away'' from the sum rate constraint. This results in a large deviations-type event for the sum rate constraint which has no bearing on second-order asymptotics; see further discussions in \cite{TK12,Nom13} and \cite{Haim12}.

Cases (ii)--(iii) are more interesting, and their proofs are non-trivial. As in Nomura-Han~\cite{Nom13} and Tan-Kosut~\cite{TK12}, the second-order asymptotics for case (ii) depend on the dispersion matrix $\bV(\rho)$ and the $2$-dimensional analogue of the  inverse of the Gaussian cdf $\Psi^{-1}$. However, in our setting, the expression containing $\Psi^{-1}$ alone (i.e. the expression obtained by setting $\beta=0$ in \eqref{eqn:second2}) corresponds to only considering the unique   input distribution $\calN(\bzero,\bSigma(\rho))$ achieving the point $(R_1^*, R_2^*)=(I_1(\rho), I_{12}(\rho)-I_1(\rho))$. As discussed in the introduction and the caption of Fig.~\ref{fig:cr}, this is  {\em not} sufficient to achieve all second-order coding rates, since there are non-empty regions within the capacity region that are not contained in the trapezoid of rate pairs achievable using $\calN(\bzero,\bSigma(\rho))$. Using a sequence of input distributions parametrized by $\rho_n$ converging to $\rho$ with rate $\Theta\big(\frac{1}{\sqrt{n}}\big)$, we obtain the Taylor expansion in~\eqref{eqn:back_envelope}, yielding the gradient term $\bD(\rho)$. 


For the converse, we consider an arbitrary sequence of codes with rate pairs $\{(R_{1,n}, R_{2,n})\}_{n\in\bbN}$ converging to $(I_1(\rho), I_{12}(\rho)-I_1(\rho))$ with second-order behavior given by \eqref{eqn:second_def}. From the global result, we know  $[R_{1,n}, R_{1,n}+R_{2,n}]^T \in \calR_{\mathrm{out}}(n,\eps;\rho_n)$ for some sequence $\{\rho_n\}$.  Combining this with the definition of the second-order coding rate in~\eqref{eqn:second_def}, we establish that $\rho_n \to \rho$.
The final result readily follows provided that $\rho_n=\rho+O\big(\frac{1}{\sqrt{n}}\big)$, and the remaining cases are shown to have no effect on $\calL$.

A similar discussion holds true for case (iii); the main differences are that the covariance matrix is singular, and that the union in \eqref{eqn:second3} is taken over $\beta\le0$ only, since $\rho_n$ can only approach one from below.

\begin{figure}
\centering
        \includegraphics[width=0.5\paperwidth]{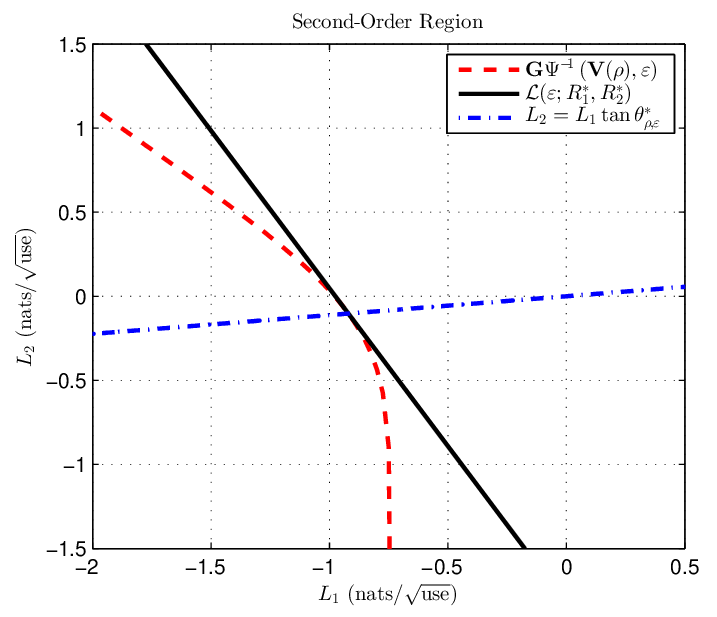} 
    \caption{Second-order coding rates in nats/$\sqrt{\mbox{use}}$  with $S_1=S_2=1$, $\rho=\frac{1}{2}$
             and $\eps = 0.1$. The regions $\bG\Psi^{-1}(\bV(\rho),\eps)$  (with $\bG := [1, 0; -1 , 1]$) and $\calL(\eps;R_1^*,R_2^*)$ are to the bottom left of the boundaries. We also plot the line $L_2=L_1\tan\theta^*_{\rho,\eps}$, where $\theta^*_{\rho,\eps}$ is the unique angle  $\theta$ for which the intersection of the boundary of $\calL(\eps;R_1^*,R_2^*)$ and the line $L_2=L_1\tan\theta$ coincides with the boundary of  $\bG\Psi^{-1}(\bV(\rho),\eps)$. \label{fig:local_region}}
\end{figure}

\subsection{Second-Order Asymptotics for a Given Angle of Approach} \label{sec:directions}
Here we study the second-order behavior when a point on the boundary is approached from a given 
angle, as was done in Tan-Kosut~\cite{TK12}.  We focus on the most interesting case in Theorem 
\ref{thm:local}, namely, case (ii) corresponding to $\rho\in(0,1)$.  Case (iii) can be handled 
similarly, and in case (i) the angle of approach is of little interest, since $L_2$ can be arbitrary. 
 
First, we present an alternative expression for the set $\calL=\calL(\eps;R_1^*,R_2^*)$ given in \eqref{eqn:second2} with $R_1^*=I_1(\rho)$ and $R_1^*+R_2^*=I_{12}(\rho)$ for some $\rho\in(0,1)$. 
It is easily seen that $(L_1,L_2)\in\calL$ implies $(L_1 + \beta D_1(\rho),L_2 + \beta D_2(\rho))\in\calL$,
where $D_2(\rho) := D_{12}(\rho) - D_1(\rho)$.  It follows
that $\calL$ equals the set of all points lying below a straight line with
slope $\frac{D_2(\rho)}{D_1(\rho)}$ which intersects the boundary of $\bG\Psi^{-1}(\bV(\rho),\eps)$, where $\bG := [1, 0; -1 , 1]$ is the invertible matrix that transforms the coordinate system from $[L_1, L_1+L_2]^T$ to $[L_1, L_2]^T$.  (In other words, $\bG\Psi^{-1}(\bV(\rho),\eps)$ is as in \eqref{eqn:second2}, but with the union removed and $\beta$ set to $0$.) 
In light of the preceding discussion,
\begin{equation}
    \calL(\eps;R_1^*,R_2^*) = \Big\{(L_1,L_2) \,:\, L_2 \le a_\rho L_1 + b_{ \rho,\eps} \Big\},
\end{equation}
where  
\begin{equation}
a_\rho := \frac{D_2(\rho)}{D_1(\rho)},\quad\mbox{and}\quad    b_{ \rho,\eps}:= \inf\Big\{b \,:\, \big(L_1,a_\rho L_1 + b\big) \in \bG\Psi^{-1}(\bV(\rho),\eps) \text{ for some } L_1 \in \bbR \Big\}.
\end{equation}
We provide an example in Fig.~\ref{fig:local_region} with the parameters $S_1=S_2=1$, $\rho=\frac{1}{2}$
and $\eps = 0.1$.  Since $\eps < \frac{1}{2}$, the boundary point $(R_1^*, R_2^*)$ is approached from the inside (see Fig.~\ref{fig:psi_set}, where for $\eps<\frac{1}{2}$, the set $\Psi^{-1}(\bV,\eps)$ only contains points with negative coordinates).

 \begin{figure}
\centering
        \includegraphics[width=0.5\paperwidth]{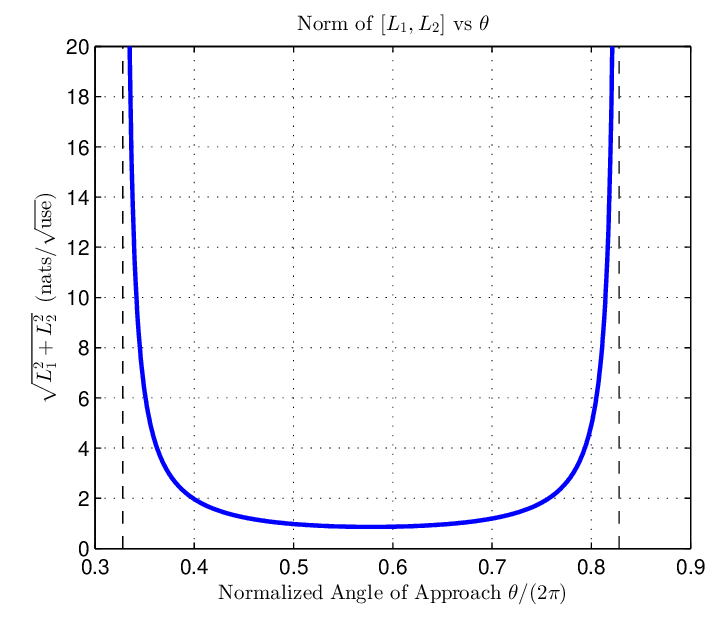} 
    \caption{Plot of $\sqrt{L_1^2+L_2^2}$ against the angle of approach $\theta$ with the same parameters as in Fig.~\ref{fig:local_region}. The second-order rates $L_1, L_2$, as functions of $\theta$, are defined in \eqref{eqn:intersection}.  Here, $\pi+\arctan(a_\rho) \approx 0.328\times 2\pi$ and $2\pi+\arctan(a_\rho) \approx 0.828\times 2\pi$ are the critical angles where  $\sqrt{L_1^2+L_2^2}$ diverges. \label{fig:local_rate}}
\end{figure}

Given the gradient $a_\rho$, the offset $b_{\rho,\eps}$, and an angle $\theta$ (measured with respect to the horizontal axis), we seek the pair $(L_1,L_2)$ on the boundary of $\calL(\eps;R_1^*,R_2^*)$ such that $L_2 = L_1\tan\theta$.  It is easily seen that this point is obtained by solving for the intersection of the line $L_2 = a_{\rho } L_1 + b_{\rho,\eps}$ with $L_2 = L_1\tan\theta$. The two lines coincide when 
\begin{align}
L_1 = \frac{b_{ \rho,\eps}}{\tan\theta-a_\rho},\quad\mbox{and}\quad L_2 = \frac{b_{ \rho,\eps}\tan\theta}{\tan\theta-a_\rho}. \label{eqn:intersection}
\end{align}
In Fig.~\ref{fig:local_region}, we see that there is only a single angle $\theta^*_{\rho,\eps}\approx  3.253\mbox{ rads}$   for which the  point of intersection in \eqref{eqn:intersection} is also on the boundary of  $\bG\Psi^{-1}(\bV(\rho),\eps)$, yielding $(L_1,L_2)\approx (-0.920,-0.103)$.  In other words, there is only one angle for which coding with a fixed input distribution   $\calN(\bzero,\bV(\rho))$ is optimal in the second-order sense (i.e.\ for which the added term $\beta \bD(\rho)$ in \eqref{eqn:second2} is of no additional help and $\beta=0$ is optimal).  For all the other angles, we should choose a non-zero coefficient $\beta$, which corresponds to choosing an input distribution that varies with $n$. 
 
Finally, in Fig.~\ref{fig:local_rate}, we plot the norm of the vector of second-order rates $[L_1, L_2]^T$ in \eqref{eqn:intersection} against $\theta$, the angle of approach. For $\eps<\frac{1}{2}$, the point  $[L_1, L_2]^T$ may be interpreted as that corresponding to the ``smallest backoff'' from the first-order  optimal rates.\footnote{There may be some imprecision in the use of the word ``backoff'' here as for angles in the second (resp.\ fourth) quadrant, $L_2$ (resp.\ $L_1$) is positive.  On the other hand, one could generally refer to ``backoff'' as moving in some inward direction relative to the capacity region boundary, even if it is in a direction where one of the second-order rates increases. The same goes for the term ``addition''. } Thus, $\sqrt{L_1^2+L_2^2}$ is a measure of the total backoff. For $\eps>\frac{1}{2}$,  $[L_1, L_2]^T$ corresponds to the ``largest addition'' to the first-order rates. It is noted that the norm tends to infinity when the angle tends to $\pi+\arctan(a_\rho)$  (from above) or  $2\pi+\arctan(a_\rho)$ (from below). This corresponds to an approach almost parallel to the gradient at the point on the boundary parametrized by $\rho$. A similar phenomenon was observed for the Slepian-Wolf problem~\cite{TK12}.

\section{Concluding Remarks}
We have identified the optimal second-order coding rate region of the Gaussian MAC with 
degraded message sets. There are two reasons as to why the analysis here is more tractable 
vis-\`a-vis  finite blocklength or second-order analysis for the the discrete memoryless MAC 
(DM-MAC) studied extensively in \cite{TK12, huang12, Mol12,  Mou13,  Scarlett13b,  Haim12}. 
Gaussianity allows us to identify the boundary of the capacity region and associate each 
point on the boundary with an input distribution parametrized by $\rho$. For the DM-MAC, 
one needs to take the convex closure of the union over input distributions $P_{X_1 , X_2}$ 
to define the capacity region~\cite[Sec.~4.5]{elgamal}, and hence the boundary points are
more difficult to characterize. In addition, one needs to ensure  in a converse proof 
(possibly related to the {\em wringing technique} of Ahlswede~\cite{Ahl82})  that the 
codewords pairs are almost orthogonal. By leveraging on the assumption of degraded message 
sets, we circumvent this requirement. 

For future investigations, we note that the Gaussian broadcast channel~\cite[Sec.~5.5]{elgamal} is a problem which is similar to the Gaussian MAC with degraded message sets (e.g. both require superposition coding, and each point on the boundary is achieved by a unique input distribution). As such, we expect that some of the  second-order analysis  techniques  contained herein may be applicable to the  Gaussian broadcast channel. The authors have recently adapted the techniques herein for the discrete memoryless MAC with degraded message sets \cite{ScaTan14}, again obtaining a conclusive characterization of the second-order rate region. 

\section{Proof of Theorem~\ref{thm:global}: Global Second-Order Result} \label{sec:prf_global}

\subsection{Converse Part}  \label{sec:prf_global_converse}
We first prove the outer bound in~\eqref{eqn:unions}.  The analysis is split into seven steps.

\subsubsection{A Reduction from Maximal to Equal Power Constraints} \label{sec:exact}

Let $\calC_{\mathrm{eq}}(n,\eps)$ be the $(n,\eps)$-capacity region in the case that~\eqref{eqn:power1} and~\eqref{eqn:power2} are equality constraints, i.e., $\|f_{1,n}(m_1, m_2)\|_2^2=nS_1$  and $\|f_{2,n}(m_2)\|_2^2=nS_2$  for all $(m_1, m_2)$.  We claim that 
\begin{align}
\calC_{\mathrm{eq}}(n,\eps)\subset\calC (n,\eps)\subset\calC_{\mathrm{eq}}(n+1,\eps). \label{eqn:eqreduction}
\end{align}
The lower bound is obvious, because the equal power constraint is more stringent than the maximal power constraint. The upper bound follows by noting that the decoder for the length-$(n+1)$ code can ignore the last symbol, which can be chosen to equalize the powers. 

It follows from \eqref{eqn:eqreduction} that for the purpose of second-order asymptotics, $\calC_{\mathrm{eq}}(n,\eps)$ and $\calC(n,\eps)$ are equivalent. This argument was also used in \cite[Lem.~39]{PPV10} and \cite[Sec.~XIII]{Sha59b}. Henceforth, we     assume that all codewords $(\bx_1, \bx_2)$ have normalized powers {\em exactly} equal to $(S_1,S_2)$.  

\subsubsection{A Reduction from Average to Maximal Error Probability} \label{sec:max_error}

Let $\calC_{\mathrm{max}}(n,\eps)$ be the $(n,\eps)$-capacity region in the
case that, along with the replacements in the previous step,
\eqref{eqn:error_prob1} is replaced by
\begin{equation}
    \max_{m_1\in[M_{1,n}],m_2\in[M_{2,n}]}\Pr\big( (\rvM_1,\rvM_2) \ne (\hat{\rvM}_1,\hat{\rvM}_2) \,\big|\, (\rvM_1,\rvM_2)=(m_1,m_2) \big)  \le\eps_n.
\end{equation}
That is, the average error probability is replaced by the maximal
error probability.  Here we show that $\calC(n,\eps)$
and $\calC_{\mathrm{max}}(n,\eps)$ are equivalent for the purposes
of second-order asymptotics, thus allowing us to focus on the maximal
error probability for the converse proof.  

By combining ideas from Csisz\'ar-K\"orner~\cite[Lem.~16.2]{Csi97} and 
Polyanskiy \cite[Sec 3.4.4]{Pol10}, we will start with
the average-error code, and use an expurgation argument to 
obtain a maximal-error code having the same asymptotic rates
and error probability.  Let $\eps_n(m_1,m_2)$ be the error 
probability given that the message pair $(m_1,m_2)$ is encoded, and let
\begin{equation}
    \eps_n(m_2) := \frac{1}{M_{1,n}}\sum_{m_1=1}^{M_{1,n}}\eps_n(m_1,m_2) \label{eq:eps_m2}
\end{equation}
be the error probability for message $m_2$, averaged over $\rvM_1$.

Consider a sequence of codes with message sets $\calM_{1,n}$ and 
$\calM_{2,n}$, having an error probability not exceeding $\eps_n$.
Let $\tilde{\calM}_{2,n}$ contain the fraction $\frac{1}{\sqrt n}$ of 
the messages $m_2\in\calM_{2,n}$ with the highest values of $\eps_n(m_2)$
(here and subsequently, we ignore rounding issues, since these do not 
affect the argument).  It follows that 
\begin{equation}
\eps_n(m_2) \le \frac{\eps_n}{1-\frac{1}{\sqrt{n}}}
\end{equation}
since otherwise the codewords not appearing in $\tilde{\calM}_{2,n}$
would contribute more than $\eps_n$ to the average error probability 
of the original code, causing a contradiction.

Before proceeding, we observe the simple fact that for each $m_2$, we can
arbitrarily re-arrange the codewords $\{\bx_1(m_1,m_2)\}_{m_1=1}^{M_{1,n}}$ (e.g.~interchanging
the codewords corresponding to two different $m_1$ values) without changing the
average or maximal error probability.  In contrast, for the standard MAC,
$\bx_1$ can only depend on $m_1$, meaning that such a re-arrangement cannot be done
\emph{separately} for each value of $m_2$.  Thus, the assumption of degraded 
message sets is crucial in the following arguments. This should be unsurprising,
since the capacity regions for the average and maximal error differ in general
for the standard MAC \cite{dueck}.

For each $m_2\in\tilde{\calM}_{2,n}$, let $\tilde{\calM}_{1,n}(m_2)$
contain the fraction $\frac{1}{\sqrt n}$ of the messages $m_1$ with the
highest values of $\eps_n(m_1,m_2)$.  By relabeling the codewords in
accordance with the previous paragraph if necessary, we can assume that
$\tilde{\calM}_{1,n} := \tilde{\calM}_{1,n}(m_2)$ is the same for each $m_2$.  
Repeating the argument following \eqref{eq:eps_m2}, we conclude that  
\begin{equation}
    \eps_n(m_1,m_2) \le \frac{\eps_n(m_2)}{1-\frac{1}{\sqrt{n}}} \le \frac{\eps_n}{\big(1-\frac{1}{\sqrt{n}}\big)^2} = \eps_n + O\left(\frac{1}{\sqrt{n}}\right) \label{eq:eps_reduction}
\end{equation}
for all $m_1\in\tilde{\calM}_{1,n}$ and $m_2\in\tilde{\calM}_{2,n}$.  Moreover,
we have by construction that 
\begin{align}
    \frac{1}{n}\log\big| \tilde{\calM}_{j,n}\big| = \frac{1}{n}\log\big| {\calM}_{j,n} \big| -  \frac{\log n}{2n} \label{eq:msg_reduction}
\end{align}
for $j=1,2$.  By absorbing the remainder terms in \eqref{eq:eps_reduction} and
\eqref{eq:msg_reduction} into the third-order term $\overg(\rho,\eps,n)$
in \eqref{eqn:Rout}, we see that it suffices to prove the converse result
for the maximal error probability. 

\subsubsection{Correlation Type Classes} \label{sec:types}

Define $\calI_0 := \{0\}$ and  $\calI_k := (\frac{k-1}{n},\frac{k}{n}], k \in [n]$, and let $\calI_{-k} := -\calI_k$ for $k\in [n]$. We see that the family $\{\calI_k:k\in [-n:n]\}$ forms a partition of $[-1, 1]$. Consider the {\em correlation type classes} (or simply {\em type classes})
\begin{align}
\calT_n(k)& := \left\{ (\bx_1, \bx_2)  : \frac{\langle\bx_1,\bx_2\rangle}{\|\bx_1\|_2 \|\bx_2\|_2} \in\calI_k \right\}  
\end{align}
where $k \in [-n: n]$, and $\langle \bx_1,\bx_2\rangle :=\sum_{i=1}^n x_{1i}x_{2i}$ is the standard inner product in $\bbR^n$.   The total number of type classes is $2n+1$, which is polynomial in $n$ analogously to the case of discrete alphabets~\cite[Ch.~2]{Csi97}.

Here we perform a further reduction (along with those in the first two steps) to codes 
for which all codeword pairs have the same type.  Let the codebook
$\calC := \{ (\bx_1(m_1, m_2),\bx_2(m_2)) : m_1 \in {\calM}_{1,n}, m_2 \in \calM_{2,n}\}$
be given; in accordance with the previous two steps, we assume that it has
codewords meeting the power constraints with equality, and \emph{maximal} 
error probability not exceeding $\eps_n$.
For each $m_2\in \calM_{2,n}$, we can find a set 
$\tilde{\calM}_{1,n}(m_2) \subset {\calM}_{1,n}$ (re-using the notation of the
previous step) such that all pairs of codewords 
$(\bx_1(m_1, m_2),\bx_2(m_2))$,  $m_1 \in \tilde{\calM}_{1,n}(m_2)$ have the same 
type, say indexed by  $k(m_2) \in [-n:n]$, and such that
\begin{align}
\frac{1}{n}\log\big| \tilde{\calM}_{1,n}(m_2)\big|\ge\frac{1}{n}\log\big|  {\calM}_{1,n}(m_2)\big|-\frac{\log (2n+1)}{n},\quad\forall\, m_2\in \calM_{2,n}. \label{eqn:reduction1}
\end{align}
We may assume that all the sets $\tilde{\calM}_{1,n}(m_2), m_2\in \calM_{2,n}$ have the same cardinality; otherwise, we can remove extra codeword pairs from some sets $\tilde{\calM}_{1,n}(m_2)$ and \eqref{eqn:reduction1} will still be satisfied. Similarly to the previous step, we may assume (by relabeling if necessary) that $\tilde{\calM}_{1,n}:=\tilde{\calM}_{1,n}(m_2)$ is the same for each $m_2$. We now have a subcodebook  $\tilde{\calC}_1:=\{(\bx_1(m_1, m_2),\bx_2(m_2)) : m_1 \in \tilde{\calM}_{1,n}, m_2 \in \calM_{2,n}\}$, where for each $m_2$, all the codeword pairs have the same type and \eqref{eqn:reduction1} is satisfied.  Across the $m_2$'s, there may be different types  indexed by $k(m_2) \in [-n:n]$, but there   exists a dominant type  indexed by $k^* \in \{k(m_2): m_2 \in \calM_{2,n}\}$ and a set $\tilde{\calM}_{2,n} \subset {\calM}_{2,n}$ such that 
\begin{align}
\frac{1}{n}\log\big| \tilde{\calM}_{2,n} \big|\ge \frac{1}{n}\log\big|  {\calM}_{2,n} \big|-\frac{\log (2n+1)}{n}. \label{eqn:reduction2}
\end{align}
As such, we have shown that  there exists a subcodebook $\tilde{\calC}_{12}:=\{(\bx_1(m_1, m_2),\bx_2(m_2)) : m_1 \in \tilde{\calM}_{1,n}, m_2 \in \tilde{\calM}_{2,n}\}$ of constant type  indexed by $k^*$  whose sum rate satisfies 
\begin{align}
\frac{1}{n}\log\big| \tilde{\calM}_{1,n}\times  \tilde{\calM}_{2,n}  \big|\ge\frac{1}{n}\log\big|  {\calM}_{1,n}\times   {\calM}_{2,n}  \big|-\frac{2\log (2n+1)}{n}. \label{eqn:reduction3}
\end{align}  
The reduced code clearly has a maximal error probability no larger than that of $\calC$.
Combining this observation with \eqref{eqn:reduction2} and \eqref{eqn:reduction3}, we see that the converse part of Theorem \ref{thm:global} for fixed-type codes implies the same for general codes, since the additional $O\big(\frac{\log n}{n})$ factors in \eqref{eqn:reduction2} and \eqref{eqn:reduction3} can be absorbed into the third-order term $\overg(\rho,\eps,n)$.  Thus, in the remainder of the proof, we limit our attention to fixed-type codes.  For each $n$, the type is indexed by $k\in[-n:n]$, and we define $\hrho := \frac{k}{n} \in [-1,1]$.  In some cases, we will be interested in \emph{sequences} of such values, in which case we will make the dependence on $n$ explicit by writing $\hrho_n$.

\subsubsection{A Verd\'u-Han-type  Converse Bound}
We  now  state a  non-asymptotic converse bound based on analogous bounds in Han's work on the information 
spectrum approach for the  general MAC  \cite[Lem.~4]{Han98} and in Boucheron-Salamatian's work on the information 
spectrum approach for the general  broadcast channel with degraded message sets~\cite[Lem.~2]{bouch00}. 
The bound only requires that the \emph{average} error probability is
no larger than $\eps_n$, which is guaranteed by the fact that the maximal 
error probability is no larger than $\eps_n$.  That is, the reduction 
to the maximal error probability in Section \ref{sec:max_error} was performed 
for the sole purpose of making the reduction to fixed types in Section 
\ref{sec:types} possible.

\begin{proposition} \label{prop:vh}
    Fix a blocklength $n\ge 1$, auxiliary  output distributions  $Q_{\bY|\bX_2}$ and $Q_{\bY }$, and a 
    constant $\gamma>0$. For any $(n,M_{1} ,M_{2} ,S_1, S_2,\eps)$-code with codewords of fixed 
    empirical powers $S_1$ and $S_2$ falling into a single correlation type class $\calT_n(k)$, there 
    exist random vectors $(\bX_1,\bX_2)$ with joint distribution $P_{\bX_1, \bX_2}$  supported on 
    $\{(\bx_1,\bx_2)\in\calT_n(k) :\|\bx_j\|_2^2 = nS_j,j=1,2\}$ such that 
    \begin{align}
    \eps\ge\Pr(\calA\cup\calB)-2e^{-n\gamma}, \label{eqn:vh}
    \end{align}
    where 
    \begin{align}
    \calA &:=\left\{ \frac{1}{n}\log\frac{W^n(\bY|\bX_1, \bX_2)}{Q_{\bY|\bX_2} (\bY|\bX_2)}\le  \frac{1}{n}\log M_{1} -\gamma\right\} \label{eqn:calAn} \\
    \calB &:=\left\{ \frac{1}{n}\log\frac{W^n(\bY|\bX_1, \bX_2)}{Q_{\bY}(\bY)} \le  \frac{1}{n}\log \big(M_{1} M_{2} \big)-\gamma\right\},\label{eqn:calBn}
    \end{align}
    with $\bY \,|\,\{ \bX_1=\bx_1,\bX_2=\bx_2 \} \sim W^n(\cdot|\bx_1,\bx_2)$.
\end{proposition}
\begin{proof}
    The proof is nearly identical to those appearing in~\cite{Han98,bouch00,Hayashi03}, so 
    we omit the details.  The starting point is the basic identity
    \begin{equation}
        \eps  \ge \Pr(\calA\cup\calB) - \Pr(\calA \cap \text{no error}) - \Pr(\calB \cap \text{no error}). \label{eq:vh_main_step}
    \end{equation} 
    We can upper bound the second probability by $e^{-n\gamma}$ by explicitly writing it
    in terms of the distributions of the codewords and the channel, and using \eqref{eqn:calAn} 
    to upper bound $W^n$ by $Q_{\bY|\bX_2} M_1 e^{-n\gamma}$.
    Handling the third term in \eqref{eq:vh_main_step} similarly yields a second 
    $e^{-n\gamma}$ term, thus resulting in \eqref{eqn:vh}.
\end{proof}

There are several differences  in Proposition~\ref{prop:vh} compared to~\cite[Lem.~4]{Han98}. First, in our work, there are constraints on the codewords, and the support of the input distribution $P_{\bX_1, \bX_2}$  is specified to reflect this. Second, there are two (instead of three) events in the probability in \eqref{eqn:vh}   because the informed encoder $f_{1,n}$ has access to both messages.   Third, we can choose arbitrary output distributions  $Q_{\bY|\bX_2}$ and $Q_{\bY}$.  This generalization is analogous  to the  non-asymptotic   converse bound by Hayashi and Nagaoka for   classical-quantum channels~\cite[Lem.~4]{Hayashi03}. The freedom to choose the output distribution is crucial in both our problem and~\cite{Hayashi03}.

\subsubsection{Evaluation of the Verd\'u-Han Bound for $\hrho\in(-1,1)$} \label{sec:eval_vh}
Recall from Sections \ref{sec:exact} and \ref{sec:types} that the codewords satisfy exact power constraints and belong to a single type class $\calT_n(k)$.  In this subsection, we consider the case that $\hrho := \frac{k}{n} \in (-1,1)$, and we derive bounds that will be useful for sequences $\hrho_n$ bounded away from $-1$ and $1$.  In Section \ref{sec:eval_vh2}, we present alternative bounds to handle the case that $\hrho_n\to\pm1$.  

We set $\gamma :=\frac{\log n}{2n}$ in \eqref{eqn:vh}, yielding $2e^{-n\gamma}=\frac{2 }{\sqrt{n}}$. 
Moreover, we choose the output distributions  $Q_{\bY|\bX_2}$ and $Q_{\bY}$ to be the $n$-fold 
products of $Q_{Y|X_2}$ and $Q_Y$, defined in \eqref{eqn:out1}--\eqref{eqn:out2} respectively, with $\hrho$ in place of $\rho$.

We now characterize the statistics of the first and second moments of 
$\sum_{i=1}^{n}\bj(x_{1i}, x_{2i}, Y_i)$ in \eqref{eqn:info_dens} for fixed sequences $(\bx_1, \bx_2)\in\calT_n(k)$.
From Appendix~\ref{sec:moments}, these moments can be expressed as
affine functions of the empirical powers $\frac{1}{n}\|\bx_1\|_2^2$, $\frac{1}{n}\|\bx_2\|_2^2$ 
and the empirical correlation coefficient $\frac{\langle\bx_1,\bx_2\rangle}{ \|\bx_1\|_2 \|\bx_2\|_2}$.
The former two quantities are fixed due to the reduction in Section \ref{sec:exact},
and the latter is within $\frac{1}{n}$ of $\hrho$ by the assumption
that $(\bx_1,\bx_2)\in\calT_n(k)$.  Moreover, a direct substitution into
\eqref{eqn:meanA} and \eqref{eqn:element_a} reveals that the mean 
vector and covariance matrix coincide with $\bI(\hrho)$ and $\bV(\hrho)$ when 
$\frac{\langle\bx_1,\bx_2\rangle}{ \|\bx_1\|_2 \|\bx_2\|_2}$ is \emph{precisely} 
equal to $\hrho$.  Combining the preceding observations, we obtain
\begin{align}
\left\|\bbE\left[\frac{1}{n}\sum_{i=1}^n \bj(x_{1i},x_{2i}, Y_i)\right] - \bI(\hrho )\right\|_\infty &\le \frac{\xi_1}{n} \label{eqn:expectation_j} \\
\left\|\cov\left[ \frac{1}{\sqrt{n}}\sum_{i=1}^n \bj(x_{1i},x_{2i}, Y_i)\right]  -  \bV(\hrho )\right\|_{\infty} &\le \frac{\xi_2}{n} \label{eqn:Vcalc}
\end{align}
for $\bY \sim W^n(\cdot|\bx_1,\bx_2)$, where $\xi_1>0$ and $\xi_2>0$ are constants.
Moreover, we can take these constants to be independent of $\hrho$, since the corresponding 
coefficients in \eqref{eqn:meanA} and \eqref{eqn:element_a} are uniformly bounded.

Let $R_{j,n} := \frac{1}{n}\log M_{j,n}$ for $j = 1,2$, and let $\bR_{n} := [R_{1,n}  , R_{1,n} + R_{2,n}  ]^T$.  We have 
\begin{align}
\Pr(\calA\cup\calB)=1-\Pr(\calA^c\cap\calB^c)=1-\bbE_{\bX_1,\bX_2} \big[ \Pr(\calA^c\cap\calB^c | \bX_1,\bX_2) \big] \label{eqn:complement}
\end{align}
and in particular, using the definition of $\bj(x_1,x_2,y)$ in \eqref{eqn:info_dens} and the fact that $Q_{\bY|\bX_2}$ and $Q_{\bY}$ are product distributions,
\begin{align}
\Pr(\calA^c\cap\calB^c | \bx_1,\bx_2)& = \Pr\left( \frac{1}{n}\sum_{i=1}^n \bj(x_{1i},x_{2i}, Y_i) > \bR_{n}  -\gamma\bone \right) \label{eqn:AB}\\
&\le \Pr\left( \frac{1}{n}\sum_{i=1}^n\Big( \bj(x_{1i},x_{2i}, Y_i)-\bbE[\bj(x_{1i},x_{2i}, Y_i)] \Big) >   \bR_{n} -\bI(\hrho ) -\gamma\bone-\frac{\xi_1}{n}\bone  \right), \label{eqn:use_norm_less} 
\end{align}
where \eqref{eqn:use_norm_less} follows from~\eqref{eqn:expectation_j}.

We are now in a position to apply the multivariate Berry-Esseen theorem~\cite{Got91, Bha10} (see Appendix~\ref{app:be}).
The first two moments are bounded according to \eqref{eqn:expectation_j}--\eqref{eqn:Vcalc},
and in Appendix \ref{sec:moments} we show that, upon replacing
the given $(\bx_1,\bx_2)$ pair  by a different pair yielding the same 
statistics of $\sum_{i=1}^{n}\bj(x_{1i},x_{2i},Y_i)$ if necessary (cf. Lemma \ref{lem:dependence}), the required 
third moment is uniformly bounded   (cf. Lemma \ref{lem:T_bd}).  It follows that
\begin{align}
&\Pr(\calA^c\cap\calB^c | \bx_1,\bx_2)  \nn\\
&\le    \Psi \Bigg(\sqrt{n} \Big( I_1(\hrho )  \! +\! \gamma \!+\!\frac{\xi_1}{n} \!-\!  R_{1,n}  \Big) ,
  \sqrt{n} \Big(  I_{12}(\hrho ) \!+ \!\gamma \!+\!\frac{\xi_1}{n}\!- \! (R_{1,n} \! +\! R_{2,n} )\Big) ; \cov\left[ \frac{1}{\sqrt{n}}\sum_{i=1}^n \bj(x_{1i},x_{2i}, Y_i)\right]   \Bigg)  \!+\! \frac{\psi(\hrho ) }{\sqrt{n}}, \label{eqn:berry}
\end{align}
where $\psi(\hrho)$ represents the remainder term. 
By Taylor expanding the continuously differentiable function $(z_1,z_2,\bV)\mapsto\Psi(z_1,z_2;\bV)$, and using the approximation in~\eqref{eqn:Vcalc} and the fact that $\det(\bV(\hrho))>0$ for $\hrho\in(-1,1)$, we obtain
\begin{align}
 \Pr(\calA^c\cap\calB^c | \bx_1,\bx_2)   \le  \Psi \left( \sqrt{n} \big(  I_1(\hrho ) -R_{1,n}   \big) , \sqrt{n} \big(  I_{12}(\hrho )   - (R_{1,n}  + R_{2,n})\big) ; \bV(\hrho )   \right) + \frac{\eta(\hrho )\log n}{\sqrt{n}} \label{eqn:taylor}
\end{align}
for some suitable remainder term $\eta(\hrho)$.
 It should be noted that $\psi(\hrho ),\eta(\hrho )\to \infty$ as $\hrho \to  \pm 1$, since $\bV(\hrho )$ becomes singular as $\hrho \to \pm 1$.  Despite this non-uniformity, we conclude from \eqref{eqn:vh}, \eqref{eqn:complement} and \eqref{eqn:taylor} that any $(n,\eps)$-code with codewords in $\calT_n(k)$ must have rates that satisfy 
\begin{align}
\begin{bmatrix}
R_{1,n}  \\ R_{1,n}  + R_{2,n} 
\end{bmatrix}& \in \bI(\hrho ) + \frac{\Psi^{-1}\Big(\bV(\hrho ),\eps + \frac{2}{\sqrt{n}} + \frac{\eta(\hrho )\log n}{\sqrt{n}}\Big)}{\sqrt n} \label{eqn:constant_type_bd0}.
\end{align}
The following ``continuity'' lemma for $\eps\mapsto\Psi^{-1}(\bV,\eps)$ is proved in Appendix \ref{app:upper_bd_set}.
\begin{lemma} \label{lem:approx}
Fix $0<\eps<1$ and a positive sequence $\lambda_n=o(1)$. Let $\bV$ be a non-zero positive semi-definite matrix. There exists a function $h(\bV,\eps )$ such that
\begin{align}
\Psi^{-1}\big( \bV,\eps+\lambda_n)\subset\Psi^{-1}\big( \bV,\eps )+h(   \bV,\eps)\,\lambda_n\, \bone, \label{eqn:inclusion}
\end{align}
and such that $h(\bV(\rho),\eps)$ is finite for each $\rho\ne\pm 1$, while being possibly divergent only as $\rho\to\pm1$. 
\end{lemma}
We conclude from Lemma~\ref{lem:approx} that
\begin{align}
\Psi^{-1}\Big(\bV(\hrho ),\eps + \frac{2}{\sqrt{n}} + \frac{\eta(\hrho )\log n}{\sqrt{n}}\Big) \subset  \Psi^{-1}\big(\bV(\hrho ),\eps  \big) +\frac{ h(\hrho,\eps) \log n}{\sqrt{n}}\bone \label{eqn:taylor_expand_set}
\end{align}
where $h(\hrho,\eps):=h(\bV(\hrho),\eps)$ diverges only as $\hrho\to\pm 1$. Uniting \eqref{eqn:constant_type_bd0} and \eqref{eqn:taylor_expand_set}, we deduce that
\begin{align}
\begin{bmatrix}
R_{1,n}  \\ R_{1,n}  + R_{2,n} 
\end{bmatrix} \in \bI(\hrho ) + \frac{\Psi^{-1}\big(\bV(\hrho ),\eps\big)}{\sqrt{n}}  + \frac{h(\hrho ,\eps)\log n}{n}\bone.\label{eqn:constant_type_bd}
\end{align}

\subsubsection{Evaluation of the Verd\'u-Han Bound with $\hrho_n\to\pm1$} \label{sec:eval_vh2}

Here we consider a sequence of codes of a single type   indexed by $k_n$ such 
that $\hrho_n := \frac{k_n}{n} \to 1$.  The case $\hrho_n\to-1$ is handled 
similarly, and the details are thus omitted.  Our aim is to show that 
\begin{align}
\begin{bmatrix}
R_{1,n}  \\ R_{1,n}  + R_{2,n} 
\end{bmatrix} \in \bI(\hrho_n) + \frac{\Psi^{-1}\big(\bV(\hrho_n),\eps\big)}{\sqrt{n}}  + o\left(\frac{1}{\sqrt n}\right)\bone. \label{eqn:rho1aim1}
\end{align}
The following lemma states that   as $\hrho_n\to 1$, the set $\Psi^{-1}(\bV(\hrho_n),\eps\big)$ in \eqref{eqn:rho1aim1} can be approximated by $\Psi^{-1}(\bV(1),\eps\big)$, which is a simpler rectangular set. The proof of the lemma is provided in Appendix~\ref{app:bounds_set}.
\begin{lemma} \label{lem:bounds_set}
Fix $0<\eps<1$ and a sequence $\{\rho_n\}$ such that $\hrho_n\to 1$.  There exist  positive sequences $a_n,b_n = \Theta( (1-\hrho_n)^{1/4})$ and $c_n= \Theta( (1-\hrho_n)^{1/2})$ satisfying
\begin{equation}
 \begin{bmatrix}
0  \\ \sqrt{V_{12}(1)}\Phi^{-1}(\eps + a_n)
\end{bmatrix}^-  -  b_n \bone \subset\Psi^{-1}(\bV(\hrho_n),\eps)\subset  \begin{bmatrix}
0  \\ \sqrt{V_{12}(1)}\Phi^{-1}(\eps  )
\end{bmatrix}^-  +  c_n \bone \label{eqn:set_up_bd} .
\end{equation}
\end{lemma}


From the inner bound in Lemma~\ref{lem:bounds_set}, in order to show \eqref{eqn:rho1aim1} it suffices to show 
\begin{align}
\begin{bmatrix} R_{1,n}  \\ R_{1,n}  + R_{2,n} \end{bmatrix} 
\le \bI(\hrho_n) + \sqrt{\frac{V_{12}(1)}{n}}\begin{bmatrix} 0 \\ \Phi^{-1}(\eps) \end{bmatrix} + o\left(\frac{1}{\sqrt n}\right)\bone , \label{eqn:rho1aim2}
\end{align}
where we absorbed the sequences $a_n ,b_n$ into the $o\big(\frac{1}{\sqrt{n}}\big)$ term.

We return to the step in \eqref{eqn:use_norm_less}, which 
when combined with the Verd\'{u}-Han-type bound in Proposition \ref{prop:vh} (with $\gamma := \frac{\log n}{2n}$) yields
for some $(\bx_1,\bx_2)\in\calT_n(k)$ that
\begin{align}
    \eps_n &\ge 1-\Pr\bigg( \frac{1}{n}\sum_{i=1}^n\Big( \bj(x_{1i},x_{2i}, Y_i)-\bbE[\bj(x_{1i},x_{2i}, Y_i)] \Big) > \bR_n - \bI(\hrho_n) - \gamma\bone-\frac{\xi_1}{n}\bone \bigg) - \frac{2}{\sqrt n} \\
           &\ge \max\Bigg\{ \Pr\bigg( \frac{1}{n}\sum_{i=1}^n\Big( j_1(x_{1i},x_{2i}, Y_i)-\bbE[j_1(x_{1i},x_{2i}, Y_i)] \Big) \le R_{1,n} - I_1(\hrho_n) - \gamma - \frac{\xi_1}{n} \bigg), \nn \\
                            &\qquad \Pr\bigg( \frac{1}{n}\sum_{i=1}^n\Big( j_{12}(x_{1i},x_{2i}, Y_i)-\bbE[j_{12}(x_{1i},x_{2i}, Y_i)] \Big) \le R_{1,n} + R_{2,n} - I_{12}(\hrho_n) - \gamma - \frac{\xi_1}{n} \bigg) \Bigg\} - \frac{2}{\sqrt n}. \label{eqn:vhweakened}
\end{align}

From \eqref{eqn:Vcalc} and the assumption that $\hrho_n \to 1$,
the variance of $\sum_{i=1}^{n}j_{12}(x_{1i},x_{2i},Y_i)$ equals 
$n(V_{12}(1)+o(1))$.  Since $V_{12}(1)>0$, we
can treat the second term in the maximum in \eqref{eqn:vhweakened} in an 
identical fashion to the single-user setting \cite{Strassen,PPV10}
to obtain the second of the element-wise inequalities in \eqref{eqn:rho1aim2}.
It remains to prove the first, i.e. to show that no $\Theta\big( \frac{1}{\sqrt n} \big)$ 
addition to $R_{1,n}$  is possible for $\eps\in(0,1)$.  

Since $V_{1}(1)=1$ and $V_1(\cdot)$ is continuous in $\rho$, we have $V_1(\hrho_n)\to0$.
Combining this observation with \eqref{eqn:Vcalc}, we conclude that the
variance of $\sum_{i=1}^n j_{1}(x_{1i},x_{2i}, Y_i)$ is $o(n)$, and we thus
have from Chebyshev's inequality that
\begin{equation}
    \Pr\bigg( \frac{1}{n}\sum_{i=1}^n\Big( j_{1}(x_{1i},x_{2i}, Y_i)-\bbE[j_{1}(x_{1i},x_{2i}, Y_i)] \Big) \le \frac{c}{\sqrt{n}} \bigg) \to 1 \label{eqn:rho1case2}
\end{equation}
for all $c>0$. Substituting \eqref{eqn:rho1case2} into \eqref{eqn:vhweakened} 
and taking $c\to0$ yields $R_{1,n} \le I_{1}(\hrho_n) + o\big(\frac{1}{\sqrt n}\big)$, as desired.

\subsubsection{Completion of the Proof}

Combining \eqref{eqn:constant_type_bd} and \eqref{eqn:rho1aim1}, we conclude
that for any sequence of codes with error probability not exceeding
$\eps\in(0,1)$, we have for some sequence $\hrho_n\in[-1,1]$ that
\begin{align}
\begin{bmatrix}
R_{1,n}  \\ R_{1,n}  + R_{2,n} 
\end{bmatrix} \in \bI(\hrho_n) + \frac{\Psi^{-1}\big(\bV(\hrho_n),\eps\big)}{\sqrt{n}}  + \overg(\hrho_n,\eps,n)\bone, \label{eqn:rho1aim}
\end{align}
where $\overg(\rho,\eps,n)$ satisfies the conditions   in the 
theorem statement.  Specifically, the first condition follows from~\eqref{eqn:constant_type_bd} (with $\overg(\rho,\eps,n) :=  h(\rho,\eps)\frac{\log n}{n}$), and the second from \eqref{eqn:rho1aim1} (with $\overg(\rho,\eps,n) =  o\big(\frac{1}{\sqrt{n}}\big)$).  This concludes the proof of the global converse.

\subsection{Direct Part} \label{sec:glo_direct}
 
We now prove the inner bound in \eqref{eqn:unions}. At a high level, we will adopt the strategy of drawing random codewords on appropriate  spheres, similarly to Polyanskiy {\em et al.} \cite[Thm.~54]{PPV10} and Tan-Tomamichel~\cite{TanTom14}.

\subsubsection{Random-Coding Ensemble}

Let $\rho\in[0,1]$ be a fixed correlation parameter. The ensemble will be defined  
in such a way that, with probability one, each codeword pair falls into the set
\begin{equation}
    \calD_n(\rho) :=  \Big\{\big(\bx_1,\bx_2\big) :  \|\bx_{1}\|_2^2 = nS_1,  \|\bx_{2}\|_2^2  =  nS_2,  \langle\bx_1,\bx_2\rangle =  n\rho\sqrt{S_1S_2}  \Big\}.\label{eqn:setD}
\end{equation}
This means that the power constraints  in \eqref{eqn:power_constraints} are satisfied with equality, and the 
empirical correlation between each codeword pair is exactly $\rho$.    We use superposition coding,
in which the codewords are generated according to
\begin{equation}
    \bigg\{\Big(\bX_2 (m_2) ,\{\bX_1(m_1,m_2)\}_{m_1=1}^{M_{1,n}}\Big)\bigg\}_{m_2=1}^{M_{2,n}} \sim\prod_{m_2=1}^{M_{2,n}}\bigg( P_{\bX_2}(\bx_2(m_2))\prod_{m_1=1}^{M_{1,n}}P_{\bX_1|\bX_2}(\bx_1(m_1,m_2)|\bx_2(m_2))  \bigg) \label{eqn:super}
\end{equation}
for codeword distributions $P_{\bX_2}$ and $P_{\bX_1|\bX_2}$. We choose the codeword distributions to be 
\begin{align}
P_{\bX_2}(\bx_2) & \propto \delta\big\{\|\bx_2\|_2^2 = nS_2\big\},\qquad\mbox{and} \label{sc_px2}\\
P_{\bX_1|\bX_2}(\bx_1|\bx_2) & \propto \delta\big\{\|\bx_1\|_2^2 = nS_1, \langle\bx_1,\bx_2\rangle=n\rho\sqrt{S_1S_2}\big\} ,\label{sc_px1} 
\end{align}
where $\delta\{\cdot\}$ is the Dirac $\delta$-function, and 
$P_{\bX}(\bx)\propto\delta\{\bx\in\calA\}$ means that $P_{\bX}(\bx)=\frac{\delta\{\bx\in\calA\}}{c}$, 
with the normalization constant $c>0$ chosen such that $\int_{\calA} P_{\bX}(\bx)\, \rmd \bx = 1$.  
In other words, each $\bX_2(m_2), m_2\in [M_{2,n}]$ is drawn uniformly from an $(n-1)$-sphere (i.e.\ an $(n-1)$-dimensional 
manifold  in $\bbR^n$) with radius $\sqrt{nS_2}$ and for each $m_2$, each $\bx_1(m_1, m_2), m_1\in [M_{1,n}]$ is drawn 
uniformly from the set of all $\bx_1$ satisfying the power and correlation coefficient constraints with equality.  
We will see that this set is in fact an $(n-2)$-sphere of radius $\sqrt{nS_1(1-\rho^2)}$, and is thus non-empty
for all $\rho\in[0,1]$.  These distributions clearly ensure that the codeword pairs belong to $\calD_n(\rho)$ with probability one. 

\subsubsection{A Feinstein-type Achievability Bound}
We now state a  non-asymptotic achievability  based on an analogous bound for the   MAC~\cite[Lem.~3]{Han98}.  This bound can be considered as a dual of Proposition~\ref{prop:vh}.  Define    
\begin{align}
P_{\bX_1|\bX_2} W^n(\by|\bx_2) &: = \int_{\bbR^n} P_{\bX_1|\bX_2}(\bx_1|\bx_2) W^n(\by|\bx_1, \bx_2)\, \rmd \bx_1 ,\\
P_{\bX_1,\bX_2} W^n(\by )  &:= \int_{\bbR^n} \int_{\bbR^n} P_{\bX_1,\bX_2}(\bx_1,\bx_2) W^n(\by|\bx_1, \bx_2)\, \rmd\bx_1\, \rmd\bx_2
\end{align}
to be output  distributions induced by a joint distribution  $P_{\bX_1, \bX_2}$ and the channel $W^n$. 
Moreover, let $\frac{\rmd P_1}{\rmd P_2}$ denote the Radom-Nikodym derivative between two 
probability distributions $P_1$ and $P_2$.
\begin{proposition} \label{prop:fein}
    Fix a blocklength $n\ge 1$, a joint distribution $P_{\bX_1, \bX_2}$ 
    such that $\|\bX_1\|_2^2\le nS_1$ and $\|\bX_2\|_2^2\le nS_2$ almost surely, auxiliary output distributions
     $Q_{\bY|\bX_2}$ and $Q_{\bY }$, a constant $\gamma>0$, and two sets $\calA_1 \subseteq \calX_2^n \times \calY^n$
     and $\calA_{12} \subseteq \calY^n$. Then there exists an 
     $(n,M_{1} ,M_{2},S_1,S_2,\eps)$-code for which 
    \begin{align}
    \eps\le\Pr(\calF\cup\calG) + \Lambda_{1}e^{-n\gamma} + \Lambda_{12}e^{-n\gamma} + \Pr\big((\bX_2,\bY) \notin \calA_1\big) + \Pr\big(\bY \notin \calA_{12}\big), \label{eqn:fein}
    \end{align}
    where 
    \begin{equation}
    \Lambda_{1} := \sup_{(\bx_2,\by)\in\calA_{1}} \frac{\rmd P_{\bX_1|\bX_2}W^n(\,\cdot\,|\bx_2)}{\rmd Q_{\bY|\bX_2}(\,\cdot\,|\bx_2)}(\by) ,\quad  \Lambda_{12}:= \sup_{\by\in\calA_{12}} \frac{\rmd P_{\bX_1,\bX_2}W^n}{\rmd Q_{\bY}}(\by), \label{eqn:K12}
    \end{equation}
    and
    \begin{align}
    \calF &:=\left\{ \frac{1}{n}\log\frac{W^n(\bY|\bX_1, \bX_2)}{Q_{\bY|\bX_2} (\bY|\bX_2)}\le  \frac{1}{n}\log M_{1} +\gamma\right\} \label{eqn:calEn} \\
    \calG &:=\left\{ \frac{1}{n}\log\frac{W^n(\bY|\bX_1, \bX_2)}{Q_{\bY}(\bY)} \le  \frac{1}{n}\log \big(M_{1} M_{2} \big)+\gamma\right\} \label{eqn:calFn}
    \end{align}
    with $\bY \,|\,\{ \bX_1=\bx_1,\bX_2=\bx_2\} \sim W^n(\cdot|\bx_1,\bx_2)$.
\end{proposition} 
\begin{proof}
    The proof is essentially identical to~\cite[Lem.~3]{Han98} (among others), so we omit the details.
    We consider superposition coding of the form given in \eqref{eqn:super}, along with
    a threshold decoder that searches for a codeword pair $(x_1,x_2)$ violating the inequalities
    in \eqref{eqn:calEn}--\eqref{eqn:calFn}.  The first term in \eqref{eqn:fein} is the probability that
    the transmitted pair fails to meet this condition. The two subsequent terms correspond to
    the probability that some incorrect pair does meet this condition, and are obtained using
    the union bound and a standard change of measure argument (e.g.~see \cite{Hayashi09}).  
    The final two terms are obtained by treating the events therein as errors (i.e.~atypical events), 
    thus permitting the restrictions to $\calA_{1}$ and $\calA_{12}$ in \eqref{eqn:K12}.
\end{proof}

The main difference between \eqref{eqn:fein} and traditional Feinstein-type threshold decoding bounds 
(e.g.~\cite[Lem.~3]{Han98}, \cite[Lem.~1]{Lan06}) is that we have the freedom to choose arbitrary output 
distributions  $Q_{\bY|\bX_2}$ and $Q_{\bY }$; this comes at the cost of introducing the multiplicative 
factors $\Lambda_{1}$ and $\Lambda_{12}$ that depend on the maximum value of the Radon-Nikodym 
derivatives in~\eqref{eqn:K12}.  Our bound in \eqref{eqn:fein} allows us to exclude ``atypical'' values of $(\bx_2,\by) \notin \calA_1$ and $\by\notin\calA_{12}$, thus facilitating the bounding of $\Lambda_{1}$ and $\Lambda_{12}$.

As with all analyses involving uniform coding on spheres \cite{TanTom14,PPV10, Mol13}, it is imperative to 
control $\Lambda_{1}$ and $\Lambda_{12}$. For this purpose, we leverage the following lemma,
which is proved in Appendix~\ref{app:rn_prf}.  For concreteness, we make the dependence of certain
quantities appearing in Proposition \ref{prop:fein} on $n$ and $\rho$ explicit, e.g. $\Lambda_1(n,\rho)$.

\begin{lemma} \label{lem:rn_bd}
    Consider the setup of Proposition~\ref{prop:fein}, where the output distributions 
    are given by $Q_{\bY|\bX_2}:=(P_{X_1|X_2} W)^n$ and $Q_{\bY }:=(P_{X_1, X_2}  W)^n$ with
    $P_{X_1, X_2} :=\calN(\bzero,\bSigma(\rho))$ (see \eqref{eqn:sigmamatrix}), and the
    joint distribution $P_{\bX_1,\bX_2}$ is described by~\eqref{sc_px2}--\eqref{sc_px1}.  
    There exist sets $\calA_{1}$ and $\calA_{12}$ (depending on $n$ and $\rho$) such that 
    \begin{gather}
        \max_{\rho\in[0,1]} \max\{ \Lambda_{1}(n,\rho),\Lambda_{12}(n,\rho) \} \le \Lambda \label{eqn:Lambda_plus} \\
        \max_{\rho\in[0,1]} \max \big\{ \Pr\big((\bX_2,\bY) \notin \calA_1(n,\rho)\big), \Pr\big(\bY \notin \calA_{12}(n,\rho)\big)\big\} \le e^{- n\psi}, \label{eqn:atypical}
    \end{gather}
    for all $n>N$, where $\Lambda <\infty$, $\psi>0$ and $N \in\bbN$ are constants not depending on $\rho$.
\end{lemma}

Note that the uniformity of \eqref{eqn:Lambda_plus}--\eqref{eqn:atypical} in $\rho$ is
crucial for handling $\rho$ varying with $n$, as is required in Theorem \ref{thm:global}. 

\subsubsection{Analysis of the Random-Coding Error Probability for $\rho_n\to\rho\in[0,1)$}\label{sec:fixed_rho}

We now use Proposition~\ref{prop:fein} with the joint distribution $P_{\bX_1, \bX_2}$ in~\eqref{sc_px2}--\eqref{sc_px1}. By construction, the probability of either codeword violating the power constraint is zero.  We choose the  output distributions $Q_{\bY|\bX_2}:=(P_{X_1|X_2} W)^n$ and $Q_{\bY }:=(P_{X_1, X_2}  W)^n$  to be of the convenient product form. By using Lemma~\ref{lem:rn_bd} and Proposition~\ref{prop:fein}, we obtain
\begin{equation}
    \eps_n \le 1 - \Pr\left( \frac{1}{n}\sum_{i=1}^n\bj (X_{1i},X_{2i},Y_i) > \bR_n + \gamma\bone\right) + 2\Lambda e^{-n\gamma} + 2e^{-n\psi}\label{eqn:feinstein} 
\end{equation}
where the information density vector $\bj(x_1, x_2,y)$ is defined with respect to $P_{X_1|X_2}W(y|x_2)$ 
and $P_{X_1,X_2}W(y)$, which coincide with $Q_{Y|X_2}$ and $Q_Y$ in \eqref{eqn:out1}--\eqref{eqn:out2}.   
Choosing $\gamma := \frac{\log n}{2n}$, we notice that the final term in \eqref{eqn:feinstein} is ${2\Lambda}/{\sqrt{n}}$.
We thus obtain
\begin{equation}
    \eps_n \le \max_{(\bx_1,\bx_2)\in\calD_n(\rho)} 1 - \Pr\left( \frac{1}{n}\sum_{i=1}^n\bj (x_{1i},x_{2i},Y_i) > \bR_n + \gamma\bone\right)  + \frac{2\Lambda}{\sqrt{n}}  + 2e^{-n\psi}. \label{eqn:feinstein2}
\end{equation}

Using the definition of $\calD_n(\rho)$ in \eqref{eqn:setD} and the expressions
for the information densities in Appendix \ref{sec:moments}, we see that the empirical mean and empirical covariance of the information densities are exactly  equal to the true mutual information vector and dispersion matrix respectively, i.e.\ 
\begin{align}
\bbE\left[ \frac{1}{n}\sum_{i=1}^n\bj (x_{1i},x_{2i},Y_i) \right] & = \bI(\rho) \label{eqn:costmoment1} ,\quad\mbox{and} \\
\cov\left[ \frac{1}{\sqrt n}\sum_{i=1}^n\bj (x_{1i},x_{2i},Y_i) \right] &=  \bV(\rho)\label{eqn:costmoment2}
\end{align}
for all $(\bx_1,\bx_2)\in\calD_n(\rho)$. These are the analogues of 
\eqref{eqn:expectation_j}--\eqref{eqn:Vcalc} in the converse proof, with the slack 
parameters $\xi_1$ and $\xi_2$ replaced by zero. 
By applying the multivariate Berry-Esseen theorem~\cite{Got91, Bha10} (see Appendix~\ref{app:be}) to \eqref{eqn:feinstein2} and performing Taylor expansions similarly to Section \ref{sec:eval_vh}, we obtain
\begin{align}
\eps_n &\le 1 - \Psi\Big( \sqrt{n}\big(  I_1(\rho)- R_{1,n}\big), \sqrt{n}\big( I_{12}(\rho) - ( R_{1,n} + R_{2,n})\big) ; \bV(\rho)\Big) + \frac{\zeta(\rho,\delta)\log n}{\sqrt{n}} ,
\end{align}
where $\zeta(\rho,\delta)$ is a function depending only on $\rho$ and $\delta$,
and diverging only as $\rho\to1$. 
By inverting the relationship between the rates and the error probability similarly to Section \ref{sec:eval_vh}, 
we obtain the desired result for any sequence $\{\rho_n\}$  converging to some $\rho\in[0,1)$, i.e. the first part of the theorem. 

\subsubsection{Analysis of the Random-Coding Error Probability for $\rho_n\to1$}

We now consider a sequence of parameters such that $\rho_n\to1$.  Similarly to
\eqref{eqn:rho1aim2}, it suffices to show the achievability of $(R_{1,n},R_{2,n})$ satisfying 
\begin{align}
\begin{bmatrix} R_{1,n}  \\ R_{1,n}  + R_{2,n} \end{bmatrix} 
\ge \bI(\rho_n) + \sqrt{\frac{V_{12}(1)}{n}}\begin{bmatrix} 0 \\ \Phi^{-1}(\eps) \end{bmatrix} + o\left(\frac{1}{\sqrt n}\right) \bone, \label{eqn:rho1aim3}
\end{align}
rather than the equivalent form given by \eqref{eqn:rho1aim}; see the outer bound in Lemma~\ref{lem:bounds_set}.

Applying the union bound to one minus the probability in~\eqref{eqn:feinstein2}, we obtain
\begin{equation}
    \eps_n \le \Pr\bigg( \frac{1}{n}\sum_{i=1}^{n} j_1(x_{1i},x_{2i},Y_i) \le R_{1,n} + \gamma \bigg) + \Pr\bigg( \frac{1}{n}\sum_{i=1}^{n} j_1(x_{1i},x_{2i},Y_i) \le R_{1,n} + R_{2,n} + \gamma \bigg) + \frac{2\Lambda}{\sqrt{n}} + 2e^{-n\psi}  \label{eqn:rho1feinstein}
\end{equation}
for some $(\bx_1,\bx_2)\in\calD_n(\rho_n)$. 
The remaining arguments are again similar to Section 
\ref{sec:eval_vh2}, so we only provide a brief outline.  We fix a small $c>0$ and choose
\begin{equation}
    R_{1,n} = I_{1}(\rho_n) - \frac{c}{\sqrt{n}} - \gamma. \label{eqn:rho1R1n}
\end{equation}
Using \eqref{eqn:costmoment1}--\eqref{eqn:costmoment2} and applying
Chebyshev's inequality similarly to \eqref{eqn:rho1case2},
we see that 
\begin{equation}
    \Pr\bigg( \frac{1}{n}\sum_{i=1}^{n} j_1(x_{1i},x_{2i},Y_i) \le R_{1,n} + \gamma \bigg) \to 0
\end{equation}
for any $c>0$ (recall that $1-\rho_n\to0$ implies $V_{1}(\rho_n)\to0$).   
Hence, and applying the {\em univariate} Berry-Esseen theorem  \cite[Sec.~XVI.5]{feller} to the second probability 
in \eqref{eqn:rho1feinstein}, we obtain \eqref{eqn:rho1aim3} and the second 
part of Theorem \ref{thm:global}.

\section{Proof of Theorem~\ref{thm:local}: Local Second-Order Result}\label{sec:prf_local}

\subsection{Converse Part} \label{sec:localconverse}
We now present the  proof of the converse part of Theorem~\ref{thm:local}.    

\subsubsection{Proof for case (i) ($\rho=0$)}

To prove the converse part for case (i), it suffices to consider the
most optimistic case, namely $M_{2,n}=1$ (i.e. no information is sent
by the uninformed user).  From the single-user dispersion result given in
\cite{Hayashi09,PPV10} (cf.~\eqref{eqn:single_user}), the number of messages for user 1
must satisfy 
\begin{equation}
\log M_{1,n} \le nI_{1}(0) + \sqrt{nV_{1}(0)}\Phi^{-1}(\eps) + o(\sqrt{n}),
\end{equation}
thus proving the converse part of \eqref{eqn:second1}.

\subsubsection{Passage to a Convergent Subsequence}

In the remainder of the proof, we consider cases (ii) and (iii).
Fix  a correlation coefficient $\rho\in(0,1]$, and consider any sequence of 
$(n,M_{1,n},M_{2,n},S_1, S_2,\eps_n)$-codes satisfying \eqref{eqn:second_def}. 
Let us consider the associated rates $\{ (R_{1,n}, R_{2,n} ) \}_{n\in\bbN}$, where 
$R_{j,n}=\frac{1}{n}\log M_{j,n}$ for $j = 1,2$. As required by Definition~\ref{def;second}, we suppose that these codes satisfy 
\begin{align}
\liminf_{n\to\infty}R_{j,n} &\ge R_j^*,    \label{eqn:first_order_opt} \\
\liminf_{n\to\infty} {\sqrt{n}}\big( R_{j,n} -  R_j^*\big )&\ge L_j,\quad j = 1,2,  \label{eqn:sec_order_opt} \\
\limsup_{n\to\infty}\eps_n &\le\eps \label{eqn:error_prob}
\end{align}
for some $(R_1^*,R_2^*)$ on the boundary parametrized by $\rho$, i.e.\ $R_1^*=I_1(\rho)$ and 
$R_1^*+R_2^*=I_{12}(\rho)$. The first-order optimality condition in \eqref{eqn:first_order_opt} 
is not  explicitly required by Definition~\ref{def;second}, but it is implied by \eqref{eqn:sec_order_opt}.
Letting $\bR_n := [ R_{1,n},  R_{1,n}+ R_{2,n}]^T$, we have from the global converse bound in~\eqref{eqn:unions} 
that there exists at a (possibly non-unique) sequence $\{\rho_n\}_{n\in\bbN}\subset [-1,1]$ such that 
\begin{equation}
    \bR_n \in \bI(\rho_n)+\frac{\Psi^{-1}(\bV(\rho_n),\eps)}{\sqrt{n}} +  \overg(\rho_n,\eps,n) \bone.
\end{equation}
Since we used the $\liminf$ for the rates and $\limsup$ for the error probability in
Definition \ref{def;second}, we may pass to a \emph{convergent} (but otherwise arbitrary) subsequence 
of $\{\rho_n\}$, say indexed by $\{n_k\}_{k\in\bbN }$.  Recalling that the $\liminf$ 
(resp. $\limsup$)  is the infimum (resp. supremum) of all subsequential limits, any 
converse result associated with this subsequence also applies to the original sequence.
Note that at least one convergent subsequence is guaranteed to exist, since $[-1,1]$ is compact.  

For the sake of clarity, we avoid explicitly writing the subscript $k$. However, 
it should be understood that asymptotic notations such as $O(\cdot)$ and $(\cdot)_n \to (\cdot)$ 
are taken with respect to the convergent subsequence.

\subsubsection{Establishing The Convergence of $\rho_n$ to $\rho$} \label{sec:establishing1}

Although $\overg(\rho_n,\eps,n)$ depends on $\rho_n$, we know from Theorem
\ref{thm:global} that it is $o\big(\frac{1}{\sqrt n}\big)$ for both
$\rho_n\to\pm1$ and $\rho_n\to\rho\in(-1,1)$.  Hence, and making use of 
the previous step, we have
\begin{equation}
    \bR_n \in \bI(\rho_n)+\frac{\Psi^{-1}(\bV(\rho_n),\eps)}{\sqrt{n}} +  o\left(\frac{1}{\sqrt n}\right)\bone. \label{eqn:uniformconverse}
\end{equation}
We claim that this result implies that $\rho_n$ converges to $\rho$. Indeed, since the boundary 
of the capacity region is curved and uniquely  parametrized by $\rho$ for $\rho\in(0,1]$, 
$\rho_n\not\to\rho$ implies for some $\delta>0$ and for all sufficiently large $n$
that either $I_1(\rho_n) \le I_1(\rho)-\delta$ or $I_{12}(\rho_n) \le I_{12}(\rho) - \delta$. 
We also have from \eqref{eqn:uniformconverse} that  $R_{1,n}\le I_1(\rho_n)+\frac{\delta}{2}$ and 
$R_{1,n}+R_{2,n}\le I_{12}(\rho_n)+\frac{\delta}{2}$ for sufficiently large $n$.   Combining
these observations, we see that 
$R_{1,n}\le I_1(\rho)-\frac{\delta}{2}$ or $R_{1,n}+R_{2,n}\le I_{12}(\rho )-\frac{\delta}{2}$ . 
This, in turn, contradicts the first-order optimality conditions in \eqref{eqn:first_order_opt}.

\subsubsection{Taylor Expansion of the Mutual Information Vector} 

Because each entry of  $\bI(\rho)$ is twice continuously differentiable, a Taylor expansion yields 
\begin{align}
\bI(\rho_n) = \bI(\rho) + \bD(\rho) (\rho_n-\rho) + O\big( (\rho_n - \rho)^2 \big) \bone, \label{eqn:taylor_expand_I}
\end{align}
where $\bD(\rho)$ is the derivative of $\bI$ defined in \eqref{eqn:derI}.   In the same way, since each entry of $\bV(\rho)$ is continuously differentiable in $\rho$, we have 
\begin{align}
\|\bV(\rho_n) - \bV(\rho) \|_{\infty}   = O( \rho_n-\rho  )  . 
\end{align}
We claim that these expansions, along with \eqref{eqn:uniformconverse}, imply that
\begin{align}
\bR_n &\in\bI(\rho) + \bD(\rho) (\rho_n-\rho) + \frac{\Psi^{-1}(\bV(\rho ) ,\eps) }{\sqrt{n}} +     \left[ o\left(\frac{1}{\sqrt n}\right) +O\big( (\rho_n - \rho)^2 \big)+ O\left(\frac{ (\rho_n-\rho)^{1/2} }{\sqrt{n}} \right) \right]\bone. \label{eqn:after_taylor}
\end{align}
The final term in the square parentheses results from  the outer bound in Lemma \ref{lem:bounds_set}
for the case $\rho=1$.  For $\rho\in(0,1)$ a standard Taylor expansion yields \eqref{eqn:after_taylor} with
the last term replaced by $O\big(\frac{\rho_n-\rho}{\sqrt{n}}\big)$, and it follows that \eqref{eqn:after_taylor}
holds for any given $\rho\in(0,1]$.

\subsubsection{Completion of the Proof for Case (ii) ($\rho\in(0,1)$)}  \label{sec:completing}

Suppose for the time being that $\rho_n-\rho = O(\frac{1}{\sqrt{n}})$, and hence 
$\tau_n:=\sqrt{n}(\rho_n-\rho)$ is a bounded sequence. By the Bolzano-Weierstrass 
theorem~\cite[Thm.\ 3.6(b)]{Rudin}, $\{\tau_n\}$ contains a convergent 
subsequence, say indexed by $\{n'_k\}$; let the limit of this subsequence be $\beta\in\bbR$. 
For the blocklengths indexed by $n'_k$, we know from \eqref{eqn:after_taylor} that 
\begin{align}
\sqrt{n'_k}\big(\bR_{n'_k}-\bI(\rho)\big)\in\beta\, \bD(\rho) + \Psi^{-1}(\bV(\rho ) ,\eps)  +o(1) \, \bone, \label{eqn:subseq}
\end{align}
where the $o(1)$ term combines the $o\big(\frac{1}{\sqrt{n}}\big)$ term  in \eqref{eqn:after_taylor} 
and the deviation $(\tau_{n'_k} - \beta)\max \{ -D_1(\rho) ,  D_{12}(\rho) \}$.  
From the second-order optimality condition in \eqref{eqn:sec_order_opt}, we know that  
{\em every} convergent subsequence of $\{R_{j,n}\}_{n\in\bbN}$ has a subsequential 
limit that satisfies $\lim_{k\to\infty} \sqrt{n_k} \big(R_{j,n_k}-R_j^*)\ge L_j$ for 
$j = 1,2$. In other words,  for all $\gamma>0$, there exist an integer  $K_1$ such that 
\begin{align}
\sqrt{n'_k} \big(R_{1,n'_k}-I_1(\rho)\big) &\ge L_1-\gamma \\
\sqrt{n'_k} \big(R_{1,n'_k}+R_{1,n'_k}-I_{12}(\rho)\big) &\ge L_1+L_2-2\gamma 
\end{align}
for all   $k\ge K_1$. Thus, we may lower bound the components in the vector on the left of 
\eqref{eqn:subseq} by $L_1-\gamma$ and $L_1 + L_2-2\gamma$. There also exists an integer $K_2$ 
such that the $o(1)$  terms are upper bounded by $\gamma$ for all $k\ge K_2$. We conclude that any 
pair of $(\eps,R_1^*, R_2^*)$-second-order achievable rate pairs $(L_1, L_2)$ must satisfy
\begin{align}
\begin{bmatrix}
L_1-2\gamma \\ L_1 + L_2 -3\gamma
\end{bmatrix} \in \bigcup_{\beta\in\bbR}\left\{  \beta\, \bD(\rho)+ \Psi^{-1}(\bV(\rho ) ,\eps) \right\}.
\end{align}
Finally, since $\gamma>0$ is arbitrary, we can take $\gamma \downarrow 0$, thus yielding
the right-hand side of \eqref{eqn:second2}.
 
To complete the proof, we must handle the case that $\rho_n - \rho$ is not $O\big(\frac{1}{\sqrt n}\big)$.
By passing to another subsequence if necessary, we may assume that 
$\rho_n - \rho = \omega\big( \frac{1}{\sqrt n} \big)$.  Roughly speaking, 
in \eqref{eqn:after_taylor},  the term $\frac{1}{\sqrt{n}}\Psi^{-1}(\bV(\rho ) ,\eps)$  
is dominated by $\bD(\rho)(\rho_n-\rho)$, and hence the second-order term scales as
 $\omega(\frac{1}{\sqrt{n}})$ instead of the desired $\Theta(\frac{1}{\sqrt{n}})$.  
 To be more precise, because 
\begin{equation}
\Psi^{-1}(\bV(\rho),\eps) \subset\begin{bmatrix}
\sqrt{V_1(\rho)}\Phi^{-1}(\eps) \\ \sqrt{V_{12}(\rho)}\Phi^{-1}(\eps)
\end{bmatrix}^- ,
\end{equation}
the bound in~\eqref{eqn:after_taylor} implies that  $\bR_n$ must satisfy
\begin{align}
\bR_n \in\bI(\rho) + \bD(\rho) (\rho_n-\rho) +  \frac{1}{\sqrt{n}}\begin{bmatrix}
\sqrt{V_1(\rho)}\Phi^{-1}(\eps) \\ \sqrt{V_{12}(\rho)}\Phi^{-1}(\eps)
\end{bmatrix}^-  + o(\rho_n-\rho)\bone.
\end{align}
Therefore, we have
\begin{align}
\bR_n  \le\bI(\rho) + \bD(\rho) (\rho_n-\rho) + o( \rho_n-\rho )\bone. \label{eqn:omega_term}
\end{align}
Since the first entry of $\bD(\rho)$ is negative and the second entry is positive, \eqref{eqn:omega_term}
implies that at least one of the two $\liminf$ values in \eqref{eqn:sec_order_opt} is equal
to $-\infty$.  That is, there are either no values of $L_1$ or no values of $L_2$ such that
the desired second-order rate conditions are satisfied.  We conclude that this case plays no 
role in the characterization of $\calL$.

\subsubsection{Completion of the Proof for Case (iii) ($\rho=1$)}

The case $\rho=1$ is handled in essentially the same way as 
$\rho\in(0,1)$, so we only state the differences.  Since
$\beta$ represents the difference between $\rho_n$ and $\rho$,
and since $\rho_n\le1$, we should only consider the case
that $\beta \le 0$.  Furthermore, for $\rho=1$ the set
$\Psi^{-1}(\bV(\rho),\eps)$ can be written in a simpler form; see
Lemma~\ref{lem:bounds_set}.  Using this form,
we readily obtain \eqref{eqn:second3}.

\subsection{Direct Part}

We obtain the local result from the global result using a similar
(yet simpler) argument to the converse part in Section~\ref{sec:localconverse}.  
For fixed $\rho\in[0,1]$ and $\beta\in\bbR$, let 
\begin{equation}
    \rho_n := \rho + \frac{\beta}{\sqrt{n}}, \label{eqn:rho_sequence}
\end{equation}
where we require $\beta \ge 0$ (resp. $\beta\le0$) when $\rho=0$ (resp. $\rho=1$).
By Theorem~\ref{thm:global} (global bound) and the definition of $R_{\mathrm{in}}(n,\eps;\rho)$ in \eqref{eqn:Rin}, 
rate pairs $(R_{1,n},R_{2,n})$ satisfying
\begin{equation}
    \bR_n \in \bI(\rho_n)+\frac{\Psi^{-1}(\bV(\rho_n),\eps)}{\sqrt{n}} +  o\left(\frac{1}{\sqrt{n}}\right)\bone \label{eqn:direct_Rn}
\end{equation}
are $(n,\eps)$-achievable.
Substituting~\eqref{eqn:rho_sequence} into~\eqref{eqn:direct_Rn} and performing Taylor expansions in 
an identical fashion to the converse part (cf.\ the argument from~\eqref{eqn:taylor_expand_I} to~\eqref{eqn:after_taylor}), we obtain  
\begin{equation}
    \bR_n \in \bI(\rho) + \frac{\beta\, \bD(\rho)}{\sqrt{n}}+\frac{\Psi^{-1}(\bV(\rho),\eps)}{\sqrt{n}} + o\left(\frac{1}{\sqrt{n}}\right)\bone. \label{eqn:direct_Rn2}
\end{equation}
We immediately obtain the desired result for case (ii) where $\rho\in [0,1)$.  We also
obtain the desired result for case (iii) where $\rho=1$  using the alternative form
of $\Psi^{-1}(\bV(1),\eps)$ (see Lemma~\ref{lem:bounds_set}), similarly
to the converse proof.

For case (i), we substitute $\rho=0$ into \eqref{eqn:dvalues0}  and~\eqref{eqn:dvalues} to obtain
$\bD(\rho)=[0 ~ D_{12}(\rho)]^T$ with $D_{12}(\rho)>0$.  Since $\beta$
can be arbitrarily large, it follows from \eqref{eqn:direct_Rn2} that $L_2$ can take any real value.
Furthermore, the set $\Psi^{-1}(\bV(0),\eps)$ contains vectors with a first entry
arbitrarily close to $\sqrt{V_1(0)}\Phi^{-1}(\eps)$ (provided that the
other entry is sufficiently negative), and we thus obtain \eqref{eqn:second1}.

\appendices
\numberwithin{equation}{section}

\section{Moments of the Information Density Vector} \label{sec:moments}
Let $\rho\in[-1,1]$ be given, and recall the definition of the information density vector in \eqref{eqn:info_dens}, and the choices of $Q_{Y|X_2}$ and $Q_Y$ in \eqref{eqn:out1}--\eqref{eqn:out2}. For a given pair of sequences $(\bx_1, \bx_2)$, form the random vector
\begin{align}
\bA_n:=\frac{1}{\sqrt{n}}\sum_{i=1}^n \bj(x_{1i},x_{2i},Y_i), \label{eqn:An}
\end{align}  
where $Y_i|\{X_{1i}=x_{1i},X_{2i}=x_{2i}\}\sim W(\cdot|x_{1i}, x_{2i})$. Define the constants $\alpha:=S_1 (1-\rho^2)$, $\vartheta := S_1 + S_2 + 2\rho\sqrt{S_1 S_2}$ and $\kappa:=\rho\sqrt{{S_1}/{S_2}}$.  Then, it can be verified that 
\begin{align}
j_1(x_1, x_2,Y) &= \frac{1}{2}\log(1+\alpha)  -\frac{Z^2}{2} + \frac{(x_1-\kappa x_2+Z)^2}{2(1+\alpha)}     = \frac{-\alpha Z^2 + 2(x_1-\kappa x_2)Z }{ 2(1+\alpha)}  + f_1(x_1,x_2), \label{eqn:j1_def}\\
j_{12}(x_1, x_2,Y) &= \frac{1}{2}\log(1+\vartheta) -\frac{Z^2}{2} + \frac{(x_1+x_2+Z)^2}{2(1+\vartheta)} = \frac{-\vartheta Z^2 + 2(x_1 + x_2 ) Z }{ 2(1+\vartheta)} +f_{12}(x_1, x_2)  ,\label{eqn:j12_def}
\end{align}
where $Z:=Y-x_1-x_2\sim\calN(0,1)$ and $f_1(x_1, x_2)$ and $f_{12}(x_1, x_2)$ are deterministic functions that will not affect the covariance matrix. Taking the expectation, we obtain
\begin{align}
\bbE\big[ j_1(x_1, x_2,Y)\big] &  =  \frac{1}{2}\log(1+\alpha) -\frac{1}{2} + \frac{1+(x_1-\kappa x_2)^2}{2(1+\alpha) } =  \frac{1}{2}\log(1+\alpha) + \frac{(x_1-\kappa x_2)^2-\alpha}{2(1+\alpha) }  \label{eqn:mean_j1} , \\
\bbE\big[ j_{12}(x_1, x_2,Y)\big]& = \frac{1}{2}\log(1+\vartheta) -\frac{1}{2} + \frac{1+(x_1+x_2)^2}{2(1+\vartheta)} = \frac{1}{2}\log(1+\vartheta) +\frac{(x_1+x_2)^2 - \vartheta}{2(1+\vartheta)} \label{eqn:mean_j2} . 
\end{align}
Setting $x_1 \leftarrow x_{1i}$, $x_2\leftarrow x_{2i}$ and $Y \leftarrow Y_i$  in \eqref{eqn:mean_j1} and \eqref{eqn:mean_j2} and summing  over all $i$, we conclude that  the    mean vector of $\bA_n$   is 
\begin{align}
\bbE \big[\bA_n\big] = \sqrt{n}\left[ 
 \rvC(\alpha) + \frac{  \|\bx_1 -   \kappa\bx_2\|_2^2  - n \alpha }{2n(1+\alpha) }   \quad\,\,
 \rvC(\vartheta) +  \frac{  \|\bx_1 + \bx_2 \|_2^2 -  n \vartheta}{2n(1+ \vartheta ) }   
\right]^T . \label{eqn:meanA}
\end{align}
 From \eqref{eqn:j1_def} and \eqref{eqn:j12_def}, we deduce that
\begin{align}
\var\big[j_1(x_1, x_2,Y) \big] &= \var\Big[ \frac{-\alpha Z^2 + 2(x_1-\kappa x_2)Z }{ 2(1+\alpha)} \Big] = \frac{\alpha^2+2(x_1-\kappa x_2)^2}{ 2(1+\alpha)^2}, \label{eqn:var1} \\
\var\big[j_{12}(x_1, x_2,Y) \big]&= \var\Big[ \frac{-\vartheta Z^2 + 2(x_1 + x_2 ) Z }{ 2(1+\vartheta)} \Big] = \frac{\vartheta^2+2(x_1+x_2)^2}{ 2(1+\vartheta)^2}, \label{eqn:var2}
\end{align}
where we have used $\var[Z^2]=2$ and $\cov[Z^2, Z]=\bbE Z^3 - (\bbE Z )(\bbE Z^2 )=0$.  
The covariance is 
\begin{align}
&\cov\big[j_1(x_1, x_2,Y), j_{12}(x_1, x_2,Y) \big]  =\cov\Big[\frac{-\alpha Z^2 + 2(x_1-\kappa x_2)Z }{ 2(1+\alpha)} , \frac{-\vartheta Z^2 + 2(x_1 + x_2 ) Z }{ 2(1+\vartheta)} \Big]\\
&=\frac{1}{4(1+\alpha)(1+\vartheta)}\Big\{ \bbE\big[ ( -\alpha Z^2 + 2(x_1-\kappa x_2 ) Z)(-\vartheta Z^2 + 2(x_1 +x_2)Z)  \big] \nn\\
&\qquad\qquad\qquad\qquad\qquad\qquad  -\bbE\big[   -\alpha Z^2 + 2(x_1-\kappa x_2 ) Z \big]\bbE\big[-\vartheta Z^2 + 2(x_1 +x_2)Z   \big] \Big\}\\
&=\frac{3\alpha\vartheta+ 4(x_1-\kappa x_2 )(x_1 +x_2) -\alpha\vartheta}{4(1+\alpha)(1+\vartheta)}=\frac{ \alpha\vartheta+ 2(x_1^2+(1-\kappa)x_1x_2-\kappa x_2^2   ) }{ 2(1+\alpha)(1+\vartheta)}. \label{eqn:cov}
\end{align}
Setting $x_1 \leftarrow x_{1i}$, $x_2\leftarrow x_{2i}$ and $Y \leftarrow Y_i$  in \eqref{eqn:var1}, \eqref{eqn:var2} and \eqref{eqn:cov} and summing  over all $i$, we conclude that  covariance  matrix  of $\bA_n$  is
\begin{align}
\cov\big[ \bA_n\big]  \!=\! \left[\begin{array}{cc}
 \dfrac{  n\alpha^2 +2 \|\bx_1 - \kappa\bx_2\|_2^2    }{2n( 1+ \alpha)^2} & \dfrac{  n\alpha\vartheta \! + \! 2( \|\bx_1\|_2^2 \!+\! ( 1\!-\! \kappa) \langle \bx_1, \bx_2 \rangle \!-\! \kappa \|\bx_2\|_2^2 ) }{2n(1+ \alpha)(1+\vartheta)}\\[2ex]
 \dfrac{  n\alpha\vartheta \! +\! 2( \|\bx_1\|_2^2 \!+\! ( 1\!-\! \kappa) \langle \bx_1, \bx_2 \rangle \!-\! \kappa \|\bx_2\|_2^2 ) }{2n(1+ \alpha)(1+\vartheta)} &    \dfrac{ n\vartheta^2 +2\|\bx_1 + \bx_2\|_2^2  }{2n ( 1+ \vartheta)^2}
\end{array} \right] . \label{eqn:element_a} 
\end{align} 

In the remainder of the section, we analyze the third absolute moments associated with $\bA_n$ appearing
in the multivariate Berry-Esseen theorem~\cite{Got91, Bha10} (see Appendix \ref{app:be}).  The following lemma
will be used to replace any given $(\bx_1,\bx_2)$ pair by an ``equivalent'' pair (in the sense
that the statistics of $\bA_n$ are unchanged) for which the
corresponding third moments have the desired behavior.  This is analogous to
Polyanskiy \emph{et al.}~\cite{PPV10}, where for  the AWGN channel, one can use  a  spherical symmetry argument to replace
any given sequence $\bx$ such that $\|\bx\|_2^2=nS$ with a fixed sequence 
$(\sqrt{S},\cdots,\sqrt{S})$. In fact, this symmetry argument has been used by many other authors including Shannon \cite{Sha59b}.

\begin{lemma} \label{lem:dependence}
    The joint distribution of $\bA_n$ depends on $(\bx_1,\bx_2)$ only through the
    powers $\|\bx_1\|_2^2$, $\|\bx_2\|_2^2$ and the inner product $\langle\bx_1, \bx_2\rangle$.
\end{lemma}
\begin{proof}
    This follows by substituting \eqref{eqn:j1_def}--\eqref{eqn:j12_def} into \eqref{eqn:An}
    and using the symmetry of the additive noise sequence $\bZ=(Z_1, \ldots,  Z_n)$.  
    For example, from \eqref{eqn:j1_def}, the first entry of $\bA_n$ can be written as
    \begin{equation}
        \frac{1}{\sqrt n}\left( \frac{n}{2}\log(1+\alpha) - \frac{1}{2}\|\bZ\|_2^2 + \frac{1}{2(1+\alpha)}\|\bx_1-\kappa\bx_2+\bZ\|_2^2 \right),
    \end{equation}
    and the desired result follows by writing 
    \begin{equation}
        \|\bx_1-\kappa\bx_2+\bZ\|^2 = \|\bx_1\|^2 + \kappa^2\|\bx_2\|^2 + \|\bZ\|^2 - 2\kappa\langle\bx_1,\bx_2\rangle + 2\langle\bx_1-\kappa\bx_2,\bZ\rangle.\label{eqn:dependence}
    \end{equation}
    Since $\bZ$ is i.i.d. Gaussian (and in particular, circularly symmetric), the 
    distribution of the final term depends on $(\bx_1,\bx_2)$ only through $\|\bx_1 - \kappa\bx_2\|$, 
    which in turn depends only on $\|\bx_1\|_2^2$, $\|\bx_2\|_2^2$ and $\langle\bx_1, \bx_2\rangle$.
\end{proof}

We now provide lemmas showing that, upon replacing a given pair $(\bx_1,\bx_2)$
with an equivalent pair using Lemma~\ref{lem:dependence} if necessary, the
corresponding third moments have the desired behavior. It will prove useful 
to work with the empirical correlation coefficient
\begin{equation}
   \rho_{\mathrm{emp}}(\bx_1, \bx_2):= \frac{\langle\bx_1, \bx_2\rangle}{\|\bx_1\|_2\|\bx_2\|_2}.
\end{equation}
It is easily seen that Lemma \ref{lem:dependence} remains true when the
inner product $\langle\bx_1,\bx_2\rangle$ is replaced by this normalized  quantity.

\begin{lemma} \label{lem:T_bd}
For any fixed $\tilde{\rho} \in [-1,1]$, $S_1>0$ and $S_2>0$,  
there exists a sequence of  pairs $(\bx_1,\bx_2)$ (indexed by increasing lengths $n$) such that $\|\bx_1\|_2^2 = nS_1$, $\|\bx_2\|_2^2 = nS_2$, 
$\rho_{\mathrm{emp}}(\bx_1,\bx_2) = \tilde{\rho}$, and
\begin{equation}
   \tilde{T}_n := \sum_{i=1}^n \bbE\left[ \Big\|  \frac{1}{\sqrt{n}} \big(\bj(x_{1i},x_{2i}, Y_i) - \bbE[\bj(x_{1i},x_{2i}, Y_i)]\big)\Big\|_2^3 \right] = O\left(\frac{1}{\sqrt{n}}\right), \label{eqn:Bn}
\end{equation}
where the $O\big(\frac{1}{\sqrt{n}}\big)$ term is uniform in $\tilde{\rho} \in [-1,1]$.
\end{lemma}
\begin{proof}
Using the fact that  $\|\bv\|_2\le  \|\bv\|_1$ and $(|a|+|b|)^3 \le 4|a|^3 + 4|b|^3$, we obtain
\begin{align}
 \tilde{T}_n &\le\sum_{i=1}^n \bbE\left[ \Big\|  \frac{1}{\sqrt{n}} \big(\bj(x_{1i},x_{2i}, Y_i) - \bbE[\bj(x_{1i},x_{2i}, Y_i)]\big)\Big\|_1^3 \right]  \label{eqn:Bn_bd1}\\
 &\le 4\sum_{i=1}^n \bbE\left[ \Big|  \frac{1}{\sqrt{n}} \big(j_1(x_{1i},x_{2i}, Y_i) - \bbE[j_1(x_{1i},x_{2i}, Y_i)]\big)\Big|^3 \right] \nn\\
 &\qquad+ 4\sum_{i=1}^n \bbE\left[ \Big|  \frac{1}{\sqrt{n}} \big(j_{12}(x_{1i},x_{2i}, Y_i) - \bbE[j_{12}(x_{1i},x_{2i}, Y_i)]\big)\Big|^3 \right]\label{eqn:Bn_bd2} .
\end{align}
We now specify $(\bx_1,\bx_2)$ whose powers and correlation match those given
in the lemma statement.  Assuming for the time being that 
$|\tilde{\rho}| \le \frac{n-1}{n}$, we choose
\begin{align}
    \bx_1 &= \big(\sqrt{S_1},\cdots,\sqrt{S_1}\big) \label{eqn:seq1} \\
    \bx_2 &= \big(\sqrt{S_2(1+\eta)},\sqrt{S_2},\cdots,\sqrt{S_2},-\sqrt{S_2(1-\eta)},-\sqrt{S_2},\cdots,-\sqrt{S_2}\big), \label{eqn:seq2}
\end{align}
where $\eta\in(-1,1)$, and $\bx_2$ contains $k\ge1$ negative entries and 
$n-k\ge1$ positive entries.  It is easily seen that $\|\bx_1\|_2^2=nS_1$ and 
$\|\bx_2\|_2^2=nS_2$, as desired.  Furthermore, we can choose $k$ and $\eta$
to obtain the desired correlation since
\begin{equation}
    \langle\bx_1,\bx_2\rangle=\big(n-2(k-1)+\sqrt{1+\eta}-\sqrt{1-\eta}\big)\sqrt{S_1S_2}, \label{eqn:innerprod}
\end{equation}
and since the range of the function $f(\eta):=\sqrt{1+\eta}-\sqrt{1-\eta}$
for $\eta\in(-1,1)$ is given by $(-\sqrt{2},\sqrt{2})$.  

Using \eqref{eqn:j1_def}--\eqref{eqn:j12_def}, it can easily be verified that the third absolute
moment of each entry of $\bj(x_1,x_2,Y)$ (i.e.\ $\bbE \big|j_1(x_1, x_2,Y)-\bbE [j_1(x_1, x_2,Y)] \big|^3$ and $\bbE \big|j_{12}(x_1, x_2,Y)-\bbE [j_{12}(x_1, x_2,Y)] \big|^3$) is bounded above by some constant  for any
$(x_1,x_2)=(\sqrt{S_1},\pm\sqrt{cS_2})$ ($c\in(0,2)$).  We thus obtain \eqref{eqn:Bn}
using \eqref{eqn:Bn_bd2}. The proof is concluded by noting that a similar argument applies for the
case $\tilde{\rho} \in (\frac{n-1}{n},1]$ by replacing \eqref{eqn:seq2} by
\begin{equation}
    \bx_2 = \big(\sqrt{S_2(1+\eta)},\sqrt{S_2(1-\eta)},\sqrt{S_2},\cdots,\sqrt{S_2}\big), \label{eqn:seq2a}
\end{equation}
and similarly (with negative entries) when $\tilde{\rho} \in [-1,\frac{n-1}{n})$.
\end{proof}

\section{A Multivariate Berry-Esseen Theorem} \label{app:be}
In this section, we state a version of the multivariate Berry-Esseen theorem~\cite{Got91, Bha10} that is suited for our needs in this paper. The following is a restatement  of Corollary 38 in \cite{WKT13}. 

\begin{theorem} \label{thm:multidimensional-berry-esseen}
Let $\bU_1, \ldots, \bU_n$ be independent, zero-mean random vectors in $\bbR^d$. Let $\hat{\bZ}_n := \frac{1}{\sqrt{n}}(\bU_1 + \cdots + \bU_n)$,
Assume $\bV:=\cov(\hat{\bZ}_n)$ is positive definite with minimum eigenvalue $\lambda_{\min}(\bV)>0$. Let $t := \frac{1}{n} \sum_{i=1}^n \bbE\big[ \| \bU_i \|_2^3\big]$ and let $\bZ$ be a zero-mean  Gaussian random vector with covariance $\bV$. Then, for all $n \in \mathbb{N}$, 
\begin{equation}
\sup_{\mathscr{C} \in \mathfrak{C}_d} \big| \Pr( \hat{\bZ}_n \in \mathscr{C} ) - \Pr ( \bZ \in \mathscr{C} ) \big| 
\le \frac{k_d\, t}{\lambda_{\min}(\bV)^{3/2} \sqrt{n}},
\end{equation}
where $\mathfrak{C}_d$ is the family of all convex, Borel measurable  subsets of $\mathbb{R}^d$,  and $k_d$ is a function only of the dimension $d$ (e.g., $k_2 = 265$).
\end{theorem}

\section{Proof of Lemma~\ref{lem:approx}} \label{app:upper_bd_set} 
Fix $(z_1, z_2)  \in\Psi^{-1}\big( \bV,\eps+\lambda_n)$ and define  $\bZ = ( Z_1, Z_2)\sim\calN(\bzero,\bV)$. Since $\Psi^{-1}\big( \bV,\eps )$ is monotonic in the sense that  $\Psi^{-1}\big( \bV,\eps ) \subset\Psi^{-1}\big( \bV,\eps' )$ for $\eps\le\eps'$,  it suffices to verify that $(z_1, z_2)$ belongs to the set on the right-hand side of \eqref{eqn:inclusion} for those $(z_1, z_2)$ on the boundary of  $\Psi^{-1}\big( \bV,\eps+\lambda_n)$. That is (cf. \eqref{eqn:psiinv}),
\begin{align}
\Pr\big( Z_1\le -z_1, Z_2\le -z_2\big)=1-(\eps+\lambda_n). \label{eqn:integral}
\end{align}
 Define  $\nu_n :=\inf\left\{ \nu>0: (-z_1-\nu, -z_2-\nu) \in \Psi^{-1} ( \bV,\eps )\right\}$.  We need to show that $\nu_n=o(1)$ is bounded above by some linear function of $\lambda_n$. By using~\eqref{eqn:integral} and the definition of $\nu_n$, we see that 
\begin{align}
\lambda_n&=\Pr\big(Z_1 \in [-z_1-\nu_n, -z_1] \cup Z_2 \in [ -z_2-\nu_n,  -z_2] \big)   \\
& \ge\max_{j=1,2}\left\{ \Phi\left(\frac{ -z_j}{\sqrt{V_{jj}}}  \right) -\Phi\left( \frac{ -z_j-\nu_n}{\sqrt{V_{jj}}}\right)   \right\} . \label{eqn:min2}
\end{align}
The assumption that $\bV$ is a non-zero positive-semidefinite matrix ensures that at least one of $V_{jj}, j=1,2$ is non-zero. We have the lower bound
\begin{align}
 \Phi\left(\frac{ -z }{\sqrt{V }}  \right) -\Phi\left( \frac{ -z -\nu_n}{\sqrt{V }}\right) \ge\frac{\nu_n}{\sqrt{V }} \min\left\{ \calN(  z;0,V  ), \calN(  z+\nu_n;0,V )  \right\}.
\end{align}
Hence, for all $n$ large enough, each of the terms in $\{\cdot\}$ in~\eqref{eqn:min2} is bounded below by $\nu_n    f(z_j,  V_{jj} )$ for $j = 1,2$ where $f(z,V):=\frac{1}{2\sqrt{V}}\calN( z;0,V)$ satisfies $\lim_{ z\to\pm\infty}f(z,V)=0$.   Hence, $\nu_n\le \lambda_n\min_{j=1,2}\{  f(z_j,  V_{jj} )^{-1}\}$. For every fixed $\eps\in (0,1)$, every $(z_1, z_2)  \in\Psi^{-1}\big( \bV,\eps+\lambda_n)$  satisfies $\min\{|z_1|, |z_2|\}<\infty$, and hence $\min_{j=1,2}\{  f(z_j,  V_{jj} )^{-1}\}$ is finite.  This concludes the proof.  
\section{Proof of Lemma~\ref{lem:bounds_set}} \label{app:bounds_set}
Recall that $\hrho_n\to 1$. We start by proving the inner bound on $\Psi^{-1}(\bV(\hrho_n),\eps)$. Let $(w_1, w_2)$ be   an arbitrary element of the left-hand-side  of \eqref{eqn:set_up_bd}, i.e.\ $w_1\le -b_n$ and $w_2\le\sqrt{V_{12}(1)}\Phi^{-1}(\eps+a_n) - b_n$.  Define the random variables  $(Z_{1,n}, Z_{2,n})\sim\calN(\bzero,\bV(\hrho_n))$ and the sequence    $b_n:=( 1-\hrho_n )^{1/4}$. Consider
\begin{align}
\Pr\big(Z_{1,n}\le -w_1, Z_{2,n}\le -w_2\big)  &\ge \Pr \Big(Z_{1,n}\le  b_n , Z_{2,n}\le -\big(\sqrt{V_{12}(1)}\Phi^{-1}(\eps+a_n) - b_n\big) \Big)\\
&\ge \Pr\Big(Z_{2,n}\le -\big(\sqrt{V_{12}(1)}\Phi^{-1}(\eps+a_n) -b_n\big) \Big)-\Pr\Big( Z_{1,n} >   b_n \Big)  \label{eqn:set_mani}\\
&= \Phi\left(\frac{-\big(\sqrt{V_{12}(1)}\Phi^{-1}(\eps+a_n) -b_n\big)  }{\sqrt{V_{12}(\hrho_n)}} \right) - \Phi\left( \frac{-b_n}{\sqrt{V_{1}( \hrho_n)}}\right) . \label{eqn:mu_n}
\end{align}
From the choice of $b_n$ and the fact that  $\sqrt{V_1(\hrho_n)}  = \Theta(\sqrt{1 -\hrho_n})$  (since $V_1(\rho)=\Theta(1-\rho)$ as $\rho\to 1$ by continuous differentiability), the argument of the second term scales as $-(1-\hrho_n)^{-1/4}$, which tends to $-\infty$. Hence,  the second term vanishes.  We may thus choose a vanishing sequence $a_n$ so that the expression in \eqref{eqn:mu_n} equals $1-\eps$. Such a choice satisfies $a_n  =\Theta( b_n)=\Theta(( 1-\hrho_n )^{1/4} )$, in accordance with the lemma statement.   From the definition in \eqref{eqn:psiinv}, we have proved that $(w_1,w_2)\in \Psi^{-1}(\bV(\hrho_n),\eps)$ for this choice  of $(a_n,b_n)$.

For the outer bound on  $\Psi^{-1}(\bV(\hrho_n),\eps)$, let $(u_1,u_2)$ be an arbitrary element of $\Psi^{-1}(\bV(\hrho_n),\eps)$. By  definition, 
\begin{equation}
\Pr(Z_{1,n}\le - u_1, Z_{2,n}\le -u_2)\ge 1-\eps,
\end{equation}
where $(Z_{1,n}, Z_{2,n})\sim\calN(\bzero,\bV(\hrho_n))$ as above. Thus, 
\begin{equation}
 1-\eps\le\Pr(  Z_{2,n}\le -u_2) =\Phi\left(\frac{-u_2}{\sqrt{V_{12}(\hrho_n)}}\right).
\end{equation}
This leads to 
\begin{equation}
u_2\le\sqrt{V_{12}(\hrho_n)}\Phi^{-1}(\eps) = \sqrt{V_{12}(1)}\Phi^{-1}(\eps) + c_n'
\end{equation}
for some $c_n'=\Theta(1-\hrho_n)$, since $\rho\mapsto\sqrt{V_{12}(\rho)}$ is continuously differentiable and its derivative does not vanish at $\rho=1$. Similarly, we have
\begin{equation}
u_1\le\sqrt{ V_1(\hrho_n)}\Phi^{-1}(\eps) = c_n''
\end{equation}
for some $c_n'' =\Theta(\sqrt{1-\hrho_n})$, since $V_1(1)=0$ and $\sqrt{V_{1}(\hrho_n)}=\Theta (\sqrt{1-\hrho_n})$.  Letting $c_n:=\max\{ |c_n'|, |c_n''|\}=\Theta(\sqrt{1-\hrho_n})$,  we deduce that $(u_1, u_2)$ belongs to the rightmost set in \eqref{eqn:set_up_bd}. This completes  the proof.

\section{Proof of Lemma~\ref{lem:rn_bd}} \label{app:rn_prf}

Throughout the proof, we use the fact that for jointly Gaussian $(X_1,X_2)$ with powers $(S_1,S_2)$
and correlation $\rho$ (i.e.~the covariance matrix given in \eqref{eqn:sigmamatrix}), we have
\begin{equation}
    X_1|\{X_2=x_2\} \sim \calN\bigg( \rho\sqrt{\frac{S_1}{S_2}}x_2,S_1(1-\rho^2) \bigg). \label{eq:cond_marginal}
\end{equation}
Several aspects of the proof are similar to Polyanskiy {\em et al.}~\cite[Lem.~61]{PPV10}
for the single-user setting, so we focus primarily on the parts that are different. 

\subsection{Upper bounding $\Lambda_1$}

A straightforward symmetry argument reveals that 
\begin{equation}
\frac{\rmd P_{\bX_1|\bX_2}W^n(\,\cdot\,|\bx_2)}{\rmd Q_{\bY|\bX_2}(\,\cdot\,|\bx_2)}(\by)
\end{equation}
is the same for all $\bx_2$ having a fixed magnitude.  Since $\|\bX_2\|_2^2=nS_2$ almost 
surely by construction, we focus on the convenient sequence
$\bx_2 = (\sqrt{nS_2},0,\dotsc,0)$.  The constraint $\langle \bx_1,\bx_2 \rangle = n\rho\sqrt{S_1S_2}$
in \eqref{sc_px1} implies that the first entry of $\bx_1$ equals $\rho\sqrt{nS_1}$ with 
probability one.  Moreover, since $\|\bX_1\|_2^2=nS_1$ almost surely, the remaining $(n-1)$ symbols
must have a total power of $nS_1(1-\rho^2)$.  Since \eqref{sc_px1} is the uniform distribution on 
the set satisfying the given conditions, we conclude that the final $(n-1)$ entries of
$\bX_1$ are uniform on the sphere of radius $\sqrt{nS_1(1-\rho^2)}$ centered at zero. 

We wish to bound the Radon-Nikodym (RN) derivative of $\bY := \bX_1 + \bx_2 + \bZ$ with respect to 
$\bY' := \bX'_1 + \bx_2 + \bZ$, where $\bX_1$ has the conditional distribution in \eqref{sc_px1},
and $\bX'_1$ is i.i.d.~on $P_{X_1|X_2}W$ given $\bx_2$ (recall the choice of
$Q_{\bY|\bX_2}$ in the lemma statement).  For notational convenience, we work with
the vectors $\bYtil := \bY-\bx_2$ and $\bYtil' := \bY'-\bx_2$, we let $\bytil$ denote
a generic sequence equaling $\by-\bx_2$, and we write $\bYtil = (\Ytil_1,\Ytil_2^n)$ to split the
first entry of $\bYtil$ from the other $(n-1)$ entries (and similarly for $\bYtil'$, $\bytil$ and $\bZ$).
Since $\bYtil$ and $\bYtil'$ are shifted versions of $\bY$ and $\bY'$, it suffices to 
bound the RN derivative associated the former sequences.
Observing that $\Ytil_1$ is independent of $\Ytil_2^n$ and similarly for $\bYtil'$, we have
\begin{equation}
    \frac{\rmd P_{\bYtil}}{\rmd P_{\bYtil'}}(\bytil) = \frac{\rmd P_{\Ytil_1}}{\rmd P_{\Ytil'_1}}(\ytil_1) \frac{\rmd P_{\Ytil_2^n}}{\rmd P_{\Ytil_2^{\prime n}}}(\ytil_2^n). \label{eqn:rn_til}
\end{equation}
By \eqref{eq:cond_marginal} and the fact that the first entry of $\bx_{2}$ equals 
$\sqrt{nS_2}$, the first RN derivative  
on the right-hand side equals the ratio of the densities $\calN\big(\rho\sqrt{nS_1},1\big)$ 
and $\calN(\rho\sqrt{nS_1},1+S_1(1-\rho^2))$.  Since the means
coincide and the former has a smaller variance, this derivative is upper bounded by its
value at the mean, which equals
\begin{equation}
\frac{\frac{1}{\sqrt{2\pi}}}{\frac{1}{\sqrt{2\pi(1+S_1(1-\rho^2))}}} = \sqrt{1+S_1(1-\rho^2)},
\end{equation}
and is thus uniformly bounded in $\rho\in[0,1]$.

We now handle the second term in \eqref{eqn:rn_til}, which is between the uniform
distribution on the sphere of radius $\sqrt{nS}$ and the $(n-1)$-fold memoryless extension of  
$\calN(0,1+S)$, where $S := S_1(1-\rho^2)$.  This is the same as the setting
of \cite[Lem.~61]{PPV10} other than two differences: (i) The block length is $n-1$
instead of $n$, so the radius $\sqrt{nS}$ is slightly larger than that
which might be expected in analogy with \cite{PPV10}, namely $\sqrt{(n-1)S}$.
(ii) We must allow for all $S \in [0,S_1]$ (to accommodate for all $\rho\in[0,1]$), 
rather than considering only a fixed  \emph{positive}
value.  Fortunately, the proof of \cite[Lem.~61]{PPV10} turns out to automatically
handle both of these issues.  Rather than repeating the proof here, we simply
outline the differences.

We first note that, as in \cite{PPV10}, we can restrict attention to sequences
$\ytil_2^n$ such that 
\begin{equation}
    (n-1)(1+S-\delta) \le \|\ytil_2^n\|_2^2 \le (n-1)(1+S+\delta) \label{eq:ytil_range}
\end{equation} 
for some $\delta \in (0,1)$, since the Chernoff bound implies that the probability
of all remaining sequences vanishes exponentially fast, explaining the exponentially 
decaying term in \eqref{eqn:atypical}.  Note that the condition
in \eqref{eq:ytil_range} corresponds to the choice of $\calA_1$ in the lemma statement;
in the more general case where $\bx_2$ may differ from $(\sqrt{nS_2},0,\dotsc,0)$, 
$\ytil_2^n$ should be replaced by the projection
of $\by$ onto the $(n-1)$-dimensional subspace orthogonal to $\bx_2$.

Next, we observe that the second term in term in \eqref{eqn:rn_til} depends on 
$\ytil_2^n$ only through its squared magnitude $r := \|\ytil_2^n\|_2^2$. Thus,
using \cite[Eqs.~(212)--(213)]{PPV10} to obtain explicit formulas for the 
densities of $\|\Ytil_2^n\|_2^2$ and $\|\Ytil_2^{\prime n}\|_2^2$, we obtain
the following analog of \cite[Eq.~(426)]{PPV10}:
\begin{equation}
    \frac{\rmd P_{\Ytil_2^n}}{\rmd P_{\Ytil_2^{\prime n}}}(\ytil_2^n) = (1+S)^{\frac{n'}{2}} \exp\Big(-n'\frac{S}{2} - r\frac{S}{2(S+1)} \Big)\big((n'+1)Sr\big)^{-\frac{1}{2}(\frac{n'}{2} - 1)} 2^{\frac{n'}{2}}\Gamma\Big(\frac{n'}{2}\Big)I_{n'/2-1}\big(\sqrt{(n'+1)Sr}\big),
\end{equation}
where $I_k(z)$ is the modified Bessel function of the first kind, $\Gamma(\cdot)$
is the Gamma function, and we have written $n' := n-1$ for the sake of ease 
of comparison with \cite[Lem.~61]{PPV10}. The desired result is now obtained as in \cite[Lem.~61]{PPV10}
by upper bounding the Gamma function and Bessel function using \cite[Eq. (428)]{PPV10}
and \cite[Eq. (430)]{PPV10} (the former of which should be combined with 
$\sinh^{-1}(z) = \log(z+\sqrt{1+z^2})$) and applying algebraic manipulations.

To gain some intuition as to why arbitrarily small values of $S$ are
permitted (which is the main difference in our analysis compared to \cite{PPV10}),
one may consider the case $S=0$, corresponding to $\rho=1$.  This
case is trivial, since it yields $\Ytil_2^n = Z_2^n$ and $\Ytil_2^{\prime n} = Z_2^n$
with probability one, thus yielding an RN derivative of one.

\subsection{Upper bounding $\Lambda_{12}$} 

Observe that, by construction in \eqref{sc_px2}--\eqref{sc_px1}, we have 
$\|\bX_1+\bX_2\|_2^2 = nS_1 + n2\rho\sqrt{S_1S_2} + nS_2$ with probability one.
Thus, by symmetry, $\bX_1 + \bX_2$ is uniform on the sphere of 
radius $\sqrt{n(S_1+S_2+2\rho\sqrt{S_1S_2})}$, and the RN derivative
we seek is \emph{identical} to that characterized in the proof of~\cite[Lem.~61]{PPV10}.
Thus, the desired result follows by choosing $\calA_{12}$ in the same
way as \cite[Eq.~(416)]{PPV10}:
\begin{equation}
    \calA_{12} = \Big\{\by \,:\, n(1+S_1+S_2+2\rho\sqrt{S_1S_2} - \delta) \le \|\by\|_2^2 \le n(1+S_1+S_2+2\rho\sqrt{S_1S_2} + \delta) \Big\}
\end{equation}
for some $\delta\in(0,1)$.

\section*{Acknowledgment}

We are grateful to Ebrahim MolavianJazi for pointing us to a minor
error in an earlier version of the paper. 

This first author has been funded in part by the European Research Council under ERC grant agreement 
259663, by the European Union's 7th Framework Programme under grant agreement 
303633, and by the Spanish Ministry of Economy and Competitiveness under grant TEC2012-38800-C03-03.
The second author has been  supported by A*STAR and NUS   grants   R-263-000-A98-750/133.

\bibliographystyle{ieeetr}
\bibliography{../isitbib}

\begin{thebibliography}{10}

\bibitem{elgamal}
A.~{El~Gamal} and Y.-H. Kim, {\em Network Information Theory}.
\newblock Cambridge, U.K.: Cambridge University Press, 2012.

\bibitem{cover72}
T.~Cover, ``Broadcast channels,'' {\em IEEE Trans. on Inf. Th.}, vol.~18,
  no.~1, pp.~2--14, 1972.

\bibitem{Hayashi08}
M.~Hayashi, ``Second-order asymptotics in fixed-length source coding and
  intrinsic randomness,'' {\em IEEE Trans. on Inf. Th.}, vol.~54, pp.~4619--37,
  Oct 2008.

\bibitem{Hayashi09}
M.~Hayashi, ``Information spectrum approach to second-order coding rate in
  channel coding,'' {\em IEEE Trans. on Inf. Th.}, vol.~55, pp.~4947--66, Nov
  2009.

\bibitem{Nom13}
R.~Nomura and T.~S. Han, ``Second-order {Slepian-Wolf} coding theorems for
  non-mixed and mixed sources,'' {\em IEEE Trans. on Inf. Th.}, vol.~60,
  pp.~5553--5572, Sep 2014.

\bibitem{Nom13b}
R.~Nomura and T.~S. Han, ``Second-order resolvability, intrinsic randomness,
  and fixed-length source coding for mixed sources: Information spectrum
  approach,'' {\em IEEE Trans. on Inf. Th.}, vol.~59, pp.~1--16, Jan 2013.

\bibitem{Strassen}
V.~Strassen, ``{Asymptotische Absch\"{a}tzungen in Shannons
  Informationstheorie},'' in {\em Trans. Third Prague Conf. Inf. Theory},
  (Prague), pp.~689--723, 1962.

\bibitem{PPV10}
Y.~Polyanskiy, H.~V. Poor, and S.~Verd\'{u}, ``Channel coding in the finite
  blocklength regime,'' {\em IEEE Trans. on Inf. Th.}, vol.~56, pp.~2307--2359,
  May 2010.

\bibitem{TanTom14}
V.~Y.~F. Tan and M.~Tomamichel, ``The third-order term in the normal
  approximation for the {AWGN} channel,'' {\em IEEE Trans.\ on Inf.\ Th.},
  vol.~61, pp.~2430--2438, May 2015.

\bibitem{Sha59b}
C.~E. Shannon, ``Probability of error for optimal codes in a {G}aussian
  channel,'' {\em Bell Systems Technical Journal}, vol.~38, pp.~611--656, 1959.

\bibitem{TK12}
V.~Y.~F. Tan and O.~Kosut, ``On the dispersions of three network information
  theory problems,'' {\em IEEE Trans. on Inf. Th.}, vol.~60, no.~2,
  pp.~881--903, 2014.

\bibitem{huang12}
Y.-W. Huang and P.~Moulin, ``Finite blocklength coding for multiple access
  channels,'' in {\em Int. Symp. Inf. Th.}, (Boston, MA), 2012.

\bibitem{Mol12}
E.~{MolavianJazi} and J.~N. Laneman, ``Simpler achievable rate regions for
  multiaccess with finite blocklength,'' in {\em Int. Symp. Inf. Th.}, (Boston,
  MA), 2012.

\bibitem{Mol12b}
E.~{MolavianJazi} and J.~N. Laneman, ``A random coding approach to {G}aussian
  multiple access channels with finite blocklength,'' in {\em Allerton
  Conference}, (Monticello, IL), 2012.

\bibitem{Mol13}
E.~{MolavianJazi} and J.~N. Laneman, ``A finite-blocklength perspective on
  {G}aussian multi-access channels,'' {\em {\tt arXiv:1309.2343 [cs.IT]}}, Sep
  2013.

\bibitem{Ver12}
S.~Verd\'u, ``Non-asymptotic achievability bounds in multiuser information
  theory,'' in {\em Allerton Conference}, (Monticello, IL), 2012.

\bibitem{Mou13}
P.~Moulin, ``A new metaconverse and outer region for finite-blocklength
  {MACs},'' in {\em Info.\ Th.\ and Applications (ITA) Workshop}, (San Diego,
  CA), 2013.

\bibitem{Scarlett13b}
J.~Scarlett, A.~Martinez, and A.~{Guill\'en i F\`abregas}, ``Second-order rate
  region of constant-composition codes for the multiple-access channel,'' {\em
  IEEE Trans.\ on Inf.\ Th.}, vol.~61, no.~1, pp.~157--172, 2015.

\bibitem{Haim12}
E.~Haim, Y.~Kochman, and U.~Erez, ``A note on the dispersion of network
  problems,'' in {\em Convention of Electrical and Electronics Engineers in
  Israel (IEEEI)}, 2012.

\bibitem{Tan_FnT}
V.~Y.~F. Tan, ``Asymptotic estimates in information theory with non-vanishing
  error probabilities,'' {\em Foundations and Trends in Communications and
  Information Theory}, vol.~11, no.~1-2, pp.~1--183, 2014.

\bibitem{BLW12}
S.~I. Bross, A.~Lapidoth, and M.~Wigger, ``Dirty-paper coding for the
  {G}aussian multiaccess channel with conferencing,'' {\em IEEE Trans. on Inf.
  Th.}, vol.~58, no.~9, pp.~5640--5668, 2012.

\bibitem{Ahl82}
R.~Ahlswede, ``{An elementary proof of the strong converse theorem for the
  multiple access channel},'' {\em {J. of Combinatorics, Information \& System
  Sciences}}, pp.~216--230, 1982.

\bibitem{ScaTan14}
J.~Scarlett and V.~Y.~F. Tan, ``Second-order asymptotics for the discrete
  memoryless {MAC} with degraded message sets,'' in {\em Intl.\ Symp.\ Info.\
  Th.}, (Hong Kong), June 2015.

\bibitem{Csi97}
I.~Csisz\'{a}r and J.~{K\"{o}rner}, {\em Information Theory: Coding Theorems
  for Discrete Memoryless Systems}.
\newblock Cambridge University Press, 2011.

\bibitem{Pol10}
Y.~Polyanskiy, {\em Channel coding: {Non}-asymptotic fundamental limits}.
\newblock PhD thesis, Princeton University, 2010.

\bibitem{dueck}
G.~Dueck, ``{Maximal error capacity regions are smaller than average error
  capacity regions for multi-user channels},'' {\em {Probl. Control Inf.
  Theory}}, vol.~7, pp.~11--19, 1978.

\bibitem{Han98}
T.~S. Han, ``An information-spectrum approach to capacity theorems for the
  general multiple-access channel,'' {\em IEEE Trans. on Inf. Th.}, vol.~44,
  pp.~2773--2795, Jul 1998.

\bibitem{bouch00}
S.~Boucheron and M.~R. Salamatian, ``About priority encoding transmission,''
  {\em IEEE Trans. on Inf. Th.}, vol.~46, no.~2, pp.~699--705, 2000.

\bibitem{Hayashi03}
M.~Hayashi and H.~Nagaoka, ``General formulas for capacity of classical-quantum
  channels,'' {\em IEEE Trans. on Inf. Th.}, vol.~49, pp.~1753--1768, Jul 2003.

\bibitem{Got91}
F.~G\"{o}tze, ``{On the rate of convergence in the multivariate CLT},'' {\em
  The Annals of Probability}, vol.~19, no.~2, pp.~721--739, 1991.

\bibitem{Bha10}
R.~Bhattacharya and S.~Holmes, ``An exposition of {G\"otze}'s estimation of the
  rate of convergence in the multivariate central limit theorem,'' tech. rep.,
  Stanford University, 2010.
\newblock {\tt arxiv:1003.4254 [math.ST]}.

\bibitem{Lan06}
J.~N. Laneman, ``On the distribution of mutual information,'' in {\em
  Information Theory and Applications Workshop}, 2006.

\bibitem{feller}
W.~Feller, {\em An Introduction to Probability Theory and Its Applications}.
\newblock John Wiley and Sons, 2nd~ed., 1971.

\bibitem{Rudin}
W.~Rudin, {\em Principles of Mathematical Analysis}.
\newblock McGraw-Hill, 1976.

\bibitem{WKT13}
S.~Watanabe, S.~Kuzuoka, and V.~Y.~F. Tan, ``Non-asymptotic and second-order
  achievability bounds for coding with side-information,'' {\em IEEE Trans.\ on
  Inf.\ Th.}, vol.~61, pp.~1574--1605, Apr 2015.

\end{thebibliography}

\begin{IEEEbiographynophoto}{Jonathan Scarlett}
(S'14-M'15) was born in Melbourne, Australia, in 1988.
In 2010, he received the B.Eng.\ degree in electrical engineering and the
B.Sci.\ degree in computer science from the University of Melbourne,
Australia. In 2011, he was a research assistant at the Department of
Electrical \& Electronic Engineering, University of Melbourne. From
October 2011 to August 2014, he was a Ph.D.\ student in the Signal
Processing and Communications Group at the University of Cambridge,
United Kingdom. He is now a post-doctoral researcher with the Laboratory
for Information and Inference Systems at the \'Ecole Polytechnique F\'ed\'erale
de Lausanne, Switzerland. His research interests are in the areas of
information theory, signal processing, and high-dimensional statistics.
He received the Poynton Cambridge Australia International Scholarship, and
the `EPFL Fellows' postdoctoral fellowship co-funded by Marie Curie.
\end{IEEEbiographynophoto}

\begin{IEEEbiographynophoto}{Vincent Y. F. Tan}
(S'07-M'11-SM'15) is an Assistant Professor in
the Department of Electrical and Computer Engineering (ECE) and the
Department of Mathematics at the National University of Singapore (NUS).
He received the B.A.\ and M.Eng.\ degrees in Electrical and Information
Sciences from Cambridge University in 2005 and the Ph.D.\ degree in
Electrical Engineering and Computer Science (EECS) from the Massachusetts
Institute of Technology in 2011. He was a postdoctoral researcher in the
Department of ECE at the University of Wisconsin-Madison and a research
scientist at the Institute for Infocomm (I$^2$R) Research, A*STAR, Singapore.
His research interests include information theory, machine learning and signal
processing.

Dr.\ Tan received the MIT EECS Jin-Au Kong outstanding doctoral thesis
prize in 2011 and the NUS Young Investigator Award in 2014. He is currently 
an Editor of the IEEE Transactions  on Communications. \end{IEEEbiographynophoto}

\end{document}